\documentclass[10pt,english]{article}
\usepackage[T1]{fontenc}
\usepackage[latin1]{inputenc}
\usepackage{graphicx}    
\usepackage{amsmath,amssymb,amsfonts,amsthm}
\usepackage{enumerate,subfigure,array}
\usepackage{float}
\usepackage[babel=true]{csquotes}
\usepackage{amsthm}
\usepackage{enumerate,subfigure,array}
\newcolumntype{g}{>{$}c<{$}}
\newtheorem{theorem}{\textbf{Theorem}}[section]
\newtheorem{lemma}{\textbf{Lemma}}[section]
\newtheorem{proposition}{\textbf{Proposition}}[section]
\newtheorem{definition}{\textbf{Definition}}[section]
\newtheorem{corollary}{\textbf{Corollary}}[section]
\newtheorem{remark}{\textbf{Remark}}[section]

\newcommand{\cF}{\mathcal{F}}

\newcommand{\R}{\mathbb{R}}

\newcommand{\E}{\mathbb{E}}

\newcommand{\N}{\mathbb{N}}

\newcommand{\Px}{\mathbb{P}}

\newcommand{\sgn}{\mathrm{sgn}}
\newcommand{\dd}{\textup{d}}

\newcommand{\phis}{\varphi_\textup{s}}
\newcommand{\phic}{\varphi_\textup{c}}
\newcommand{\ios}{\iota_\textup{s}}
\newcommand{\ioc}{\iota_\textup{c}}

\begin{document}

\title{Dynamic optimal execution in a mixed-market-impact Hawkes price model
}
\author{Aur\'elien Alfonsi, Pierre Blanc\thanks{Universit\'e Paris-Est, CERMICS, Projet MATHRISK
    ENPC-INRIA-UMLV, 6 et 8 avenue Blaise Pascal, 77455 Marne La Vall\'ee, Cedex
    2, France, e-mails : alfonsi@cermics.enpc.fr, blancp@cermics.enpc.fr. P. Blanc is grateful to Fondation Natixis for his Ph.D. grant. This research also benefited
    from the support of the ``Chaire Risques Financiers'', Fondation du Risque.
}} 

\date{\today}
\maketitle

 \begin{abstract}
We study a linear price impact model including other liquidity takers, whose flow of orders follows a Hawkes process. The optimal execution problem is solved explicitly in this context, and the closed-formula optimal strategy describes in particular how one should react to the orders of other traders. This result enables us to discuss the viability of the market. It is shown that Poissonian arrivals of orders lead to quite robust Price Manipulation Strategies in the sense of Huberman and Stanzl~\cite{HS}. Instead, a particular set of conditions on the Hawkes model balances the self-excitation of the order flow with the resilience of the price, excludes Price Manipulation Strategies and gives some market stability.


{\bf AMS (2010)}: 91G99, 91B24, 91B26, 60G55, 49J15. 


\end{abstract}


\vspace{20mm}

\section{Introduction}

When modeling the price of an asset, we typically distinguish at least three different time scales. At the low-frequency level, the price can often be well approximated by a diffusive process. At the other end, when dealing with very high frequencies, some key features of the Limit Order Book (LOB) dynamics have to be modeled. In between, price impact models consider an intra-day mesoscopic time scale, somewhere between seconds and hours. They usually ignore most of the LOB events (limit orders, cancellations, market orders, etc.) and focus on describing the price impact of the transactions. Their goal is to be more tractable than high-frequency models and to bring quantitative results on practical issues such as optimal execution strategies. 
The usual setup is well-described in Gatheral~\cite{Gatheral}, who defines the price process $S$ as
 \begin{equation}
S_t \ = \ S_0 + \int_0^t f( \dot x_s) G(t - s) \textup{d}s + \int_0^t \sigma dZ_s,
\nonumber
\end{equation}
where $\dot x_s$ is the rate of trading of the liquidating agent at time $s < t$, $f(v)$ represents the instantaneous price
impact of an agent trading at speed $v$, $G$ is called a \enquote{decay kernel} and $Z$ is a noise process. The quantity $f(v) G(+\infty)$ is usually called the \enquote{permanent impact}, $f(v) G(0)$ the \enquote{immediate impact} and $f(v) [G(0^+) - G(+\infty)]$ the \enquote{transient impact}.
The pioneering price impact models of Bertsimas and Lo~\cite{BL} and Almgren and Chriss~\cite{AC} consider a linear impact, with an immediate and a permanent part (which corresponds to $f(v) = \alpha v, \ G(0)>0, \ G(0^+) = G(+\infty) > 0$ with the previous notations). These models ignore the transient part of the impact which is due to the resilience of the market and cannot be neglected when trading frequently. For that purpose, Obizhaeva and Wang~\cite{OW} have considered a model that includes in addition a linear transient impact that decays exponentially (i.e. $f(v) = \alpha v, \ G(u)=\lambda + (1-\lambda) \exp(-\rho u), \ 0\leq\lambda\leq1,\ \rho>0$). However, empirical evidence on market data shows that the price impact is not linear but rather concave, see e.g. Potters and Bouchaud~\cite{PB}, Eisler et al.~\cite{EBK}, Mastromatteo, T\'oth and Bouchaud~\cite{MTB}, Donier~\cite{Donier}
and more recently, Farmer, Gerig, Lillo and Waelbroeck~\cite{FGLW}.
 Extensions or alternatives to the Obizhaeva and Wang model that include non-linear price impact 
have been proposed  by Alfonsi, Fruth and Schied~\cite{AFS}, Predoiu, Shaikhet and Shreve~\cite{PSS}, Gatheral~\cite{Gatheral} and Gu\'eant~\cite{Gueant} to mention a few. Similarly, the exponential decay of the transient impact is not truly observed on market data, and one should consider more general decay kernels. Alfonsi, Schied and Slynko~\cite{ASS} and Gatheral, Schied and Slynko~\cite{GSS} consider the extension of the Obizhaeva and Wang model when the transient impact has a general decay kernel.  
Another simplification made by these models is that they generally assume that when the liquidating trader is passive, the price moves according to a continuous martingale, that sums up the impact of all the orders issued by other participants. 
However, if one wants to use these models at a higher frequency, they would naturally wonder how these orders (at least the largest ones) can be taken into account in the strategy, and if the martingale hypothesis for the price can be relaxed. This is one of the contributions of the present paper. 

On the other hand, high-frequency price models aim at reproducing some statistical observations made on market data such as the autocorrelation in the signs of trades, the volatility clustering effect, the high-frequency resilience of the price, etc., and to obtain low-frequency asymptotics that are consistent with continuous diffusions. At very high frequencies, one then has to describe LOB dynamics, or a part of it. Such models have been proposed by  Abergel and Jedidi~\cite{AbergelJedidi}, Huang, Lehalle and Rosenbaum~\cite{HLR}, Cont and de Larrard~\cite{ContLarrard}, Gar\`eche et al.~\cite{GDKB}, among others. However, as stressed in~\cite{ContLarrard}, LOB events are much more frequent than price moves. Thus, it may be relevant to model the price at the slightly lower frequency of midpoint price changes. For example, Robert and Rosenbaum~\cite{RobertRosenbaum} have proposed a model based on a diffusion with uncertainty zones that trigger the price changes. Recently, Bacry et al.~\cite{BDHM} presented a tick-by-tick price model based on Hawkes processes, that reproduces well some empirical facts of market data. This model has then been enriched by Bacry and Muzy~\cite{BacryMuzy} to describe jointly the order flow and the price moves. In fact, there is a very recent and active literature that focuses on the use of  mutually exciting Hawkes processes in high-frequency price models. Without being exhaustive, we mention here the works of Da Fonseca and Zaatour~\cite{FonsecaZaatour}, Zheng, Roueff and Abergel~\cite{ZRA}, Filimonov and Sornette~\cite{FilimonovSornette} and Hardiman, Bercot and Bouchaud~\cite{HBB}. Asymptotic and low-frequency behaviour of such models has been investigated recently by Bacry et al.~\cite{BDHM2} and Jaisson and Rosenbaum~\cite{JaissonRosenbaum}. 

The present paper is a contribution to this also mutually exciting literature. Its main goal is to make a bridge between high-frequency price models and optimal execution frameworks. On the one hand, Hawkes processes seem to be rich enough to describe satisfactorily the flow of market orders. On the other hand, price impact models are tractable and well-designed to calculate trading costs. The aim of our model is to grasp these two features. Thus, we consider an Obizhaeva and Wang framework where market buy and sell orders issued by other traders are modeled through Hawkes processes. This enables us to make quantitative calculations and to solve the optimal execution problem explicitly. We obtain a necessary and sufficient condition on the parameters of the Hawkes model to rule out Price Manipulation Strategies that can be seen as high-frequency arbitrages. Interestingly, we also show that modeling the order flow with a Poisson process necessarily leads to those arbitrages. 

The paper is organized as follows. In Section~\ref{section:model_setup}, we set up the model and present a general criterion to exclude Price Manipulation Strategies. 
Section~\ref{section:mr} summarizes our main results.
Section~\ref{section:opt_strat} gives the solution of the optimal execution problem along with several comments and insights on the optimal strategy.
Eventually, we analyze the existence of Price Manipulation Strategies in our model in Section~\ref{section:PMS} and give the conditions under which they are impossible.
Cumbersome explicit formulas and technical proofs are gathered in the Appendix.

\section{Model setup and the optimal execution problem}\label{section:model_setup}
\subsection{General price model}
We start by describing the price model itself, without considering the execution problem. We consider a single asset and denote by $P_t$ its price at time $t$. We assume that we can write it as the sum of a \enquote{fundamental price} component $S_t$ and a \enquote{mesoscopic price deviation} $D_t$:
\begin{equation}\label{defP}
P_t  \ = \underset{\text{fundamental price}}{\underbrace{S_t}}
\ + \ \underset{\text{mesoscopic price deviation}}{\underbrace{D_t}.}
\end{equation}
Typically, these quantities are respectively related to the permanent and the transient impact of the market orders. We now specify the model and consider the framework of Obizhaeva and Wang~\cite{OW} where these impacts are linear. Let $N_t$ be the sum of the signed volumes of past market orders on the book between time~$0$ and time~$t$. By convention, a buy order is counted positively in~$N$ while a sell order makes~$N$ decrease, and we assume besides that $N$ is a  c\`adl\`ag (right continuous with left limits) process.  We assume that an order modifies the price proportionally to its size, which would correspond to a block-shaped limit order book. A proportion $\nu \in [0, 1]$ of the price impact is permanent, while the remaining proportion $1-\nu$ is transient with an exponential decay of speed $\rho > 0$. This mean-reversion effect can be seen as the feedback of market makers, who affect the price using limit orders and cancellations. Namely, we consider the following dynamics for~$S$ and $D$:
\begin{align*}
\textup{d}S_t  & \quad = \quad  \frac \nu q \ \underset{\text{market orders}}{\underbrace{\textup{d}N_t}} 		\\
\textup{d}D_t  & \quad = \quad \underset{\text{market resilience}}{\underbrace{- \rho \ D_t \ \textup{d}t}} 
\ + \ \frac{1 - \nu}q \underset{\text{market orders}}{\underbrace{\textup{d}N_t},} 	
\end{align*}
with $q>0$. 
One should note that in this model, the variations in the fundamental value of the asset are revealed in its price through the process $S$. Indeed, we assume that the impact of each incoming market order, modeled through the process $N$, contains of proportion $\nu$ of \enquote{real} or \enquote{exogenous} information, and that the remaining proportion $1-\nu$ is of endogenous origin and will vanish over time.
\begin{remark}\label{rmk:unrealistic}
This model assumes a linear price impact with an exponential resilience. As mentioned in the introduction, these assumptions are challenged by empirical facts, and it would be for sure interesting and relevant to enrich the model by considering a non linear price impact and a more general decay of the impact. However, the new feature of the model with respect to the literature on optimal execution is to add a flow of market orders issued by other traders. This is why we afford to make these simplifying assumptions that give analytical tractability, which is important to calculate the optimal execution strategy in real time.  Thus, the model is meant to constitute a first step in dynamic optimal execution with the price driven by point processes, and we plan to confront it to market data in a future work.
\end{remark}
As usual, we consider $(\Omega,\cF,\Px)$ a probability space where $\Px$ weights the probability of the market events. We assume that the process $(N_t)_{t\ge 0}$ has bounded variation and is square integrable, i.e. $\sup_{s\in [0,t]}\E[N_s^2]<\infty$ for any $t\ge 0$, and we define $(\cF_t)_{t \geq 0}$ the natural filtration of $N$, $\cF_t = \sigma(N_s, s\leq t)$ for $t \geq 0$. We will specify in Section~\ref{section:MIH} which dynamics we consider for~$N$ in this paper.

\subsection{Optimal execution framework}
We now consider a particular trader who wants to buy or sell a given quantity of assets on the time interval~$[0,T]$. Through the paper, we will call this trader the ``strategic trader'' to make the distinction between his market orders and all the other market orders, that are described by~$N$. We will denote by $X_t$ the number of assets owned by the strategic trader at time~$t$. We assume that the process is $(\cF_t)$-adapted, with bounded variation and c\`agl\`ad (left continuous with right limits) which means that the strategic trader observes all the information available on the market, and that he can react instantly to the market orders issued by other traders. Besides, a  strategy that liquidates $x_0$ assets on~$[0,T]$ should satisfy $X_0=x_0$ and $X_{T+}=0$: $x_0>0$ (resp. $x_0<0$) corresponds to to a sell (resp. buy) program. 
\begin{definition}\label{def:liq_strat}
A liquidating strategy~$X$ for the position $x_0\in\R$ on $[0,T]$ is admissible if it is  $(\mathcal{F}_t)$-adapted, c\`agl\`ad, square integrable, with bounded variation and such that $X_0 = x_0$ and $X_{T^+} = 0$, a.s. 
\end{definition}
\begin{remark}
An admissible strategy $X$ has a countable set $\mathcal{D}_X$ of times of discontinuity on $[0,T]$, and can have a non-zero continuous part $X^\textup{c}_t = X_t - \underset{\tau \in \mathcal{D}_X \cap [0,t)}{\sum} (X_{\tau^+}-X_\tau), \ t \in [0,T]$.
\end{remark}
One then has to specify how the strategic trader modifies the price, as well as the cost induced by his trading strategy. Again, we will consider the Obizhaeva and Wang model~\cite{OW} with the same price impact as above. However, we let the possibility that the proportion $\epsilon \in [0,1]$ of permanent impact of the strategic trader could be different from the one of the other traders, which we note $\nu \in [0,1]$. 
Of course, a reasonable choice would be to set $\epsilon = \nu$ to consider all orders equally, but the model allows for more generality.
We then assume the following dynamics
\begin{align}
\textup{d}S_t  & \quad = \quad  \frac 1 q \left( \nu  \textup{d}N_t +\epsilon \textup{d}X_t \right), 	\label{dynS}	\\
\textup{d}D_t  & \quad = - \rho \ D_t \ \textup{d}t
 +  \frac 1 q \left( (1-\nu)  \textup{d}N_t + (1-\epsilon) \textup{d}X_t \right). \label{dynD}		
\end{align}
With the assumptions on~$N$ and $X$, the price processes $P$, $S$ and $D$ have left and right limits. More precisely, in case of discontinuity at time~$t$,~\eqref{dynS} and~\eqref{dynD} have to be read here as follows
\begin{align*}
S_t-S_{t-}  & =  \frac \nu q   (N_t-N_{t-}), \  S_{t+}-S_{t}   =  \frac \epsilon q (X_{t+}-X_t),	\\
D_t-D_{t-}  &= \frac{1-\nu} q   (N_t-N_{t-}) , \  D_{t+}-D_{t}   = \frac{1-\epsilon} q (X_{t+}-X_t).
\end{align*}
For the sake of tractability only, we make the assumption of a block-shaped Limit Order Book. 
Thus (see~\cite{OW}), when the strategic trader places at  time $t$ an order of size~$v\in\R$ ($v>0$ for a buy order and $v<0$ for a sell order), it has the following cost
$$\pi_t(v) \quad = \quad \int_0^v \left[ P_t+ \frac{1}q y \right] \ \textup{d}y \quad = \quad \underset{\text{cost at the current price}}{\underbrace{P_t \ v}}  \ + \ \underset{\text{impact cost}}{\underbrace{\frac {v^2} {2q}}}. $$
Since $P_{t+}=P_t+\frac{v}{q}$, this cost amounts to trade all the assets at the average price $(P_t+P_{t+})/2$. We stress here that if an order has just occurred, i.e. $N_t-N_{t-}\not =0$, the value of~$P_t$ is different from $P_{t-}$ and takes into account the price impact of this order.  
Therefore, the cost of an admissible strategy~$X$ is given by
\begin{align}\label{CostX}
C(X)&=\int_{[0,T)} P_u \ \textup{d}X_u 
\ + \ \frac1{2q} \underset{\tau \in \mathcal{D}_X \cap [0,T)}{\sum} (\Delta X_\tau)^2
 \ - \ P_T X_T \ + \ \frac1{2q} \ X_T^2 \\
&=\int_{[0,T)} P_u \ \textup{d}X^c_u 
\ + \  \underset{\tau \in \mathcal{D}_X \cap [0,T)}{\sum} P_{\tau} (\Delta X_\tau) + \ \frac1{2q} \underset{\tau \in \mathcal{D}_X \cap [0,T)}{\sum} (\Delta X_\tau)^2
 \ - \ P_T X_T \ + \ \frac1{2q} \ X_T^2, \nonumber
\end{align}
since at time~$T$ all the remaining assets have to be liquidated. Here, the sum brings on the countable times of discontinuity $\mathcal{D}_X$ of~$X$, and the jumps $\Delta X_\tau=X_{\tau+}-X_\tau \neq 0$ for $\tau \in \mathcal{D}_X$. We note that all the terms involved in the cost function are integrable, thanks to the assumption on the square integrability of~$X$ and~$N$.
\begin{remark}\label{rk:add_liq_cost}
With the initial market price $P_0$ taken as a reference, $-P_0 \times x_0$ is the mark-to-market liquidation cost. Thus, $C(X)+P_0 \times x_0$ can be seen as an additional liquidity cost of it is positive. If it is negative, its absolute value can be seen as the gain associated to the strategy~$X$.
\end{remark}
\begin{remark}\label{Rk_Cost}
The cost defined by~\eqref{CostX} in the price model~\eqref{defP}, \eqref{dynS} and~\eqref{dynD} is a deterministic function of $(X_t)_{t\in[0,T]}$, $(N_t)_{t\in[0,T]}$, $S_0$, $D_0$ and the parameters~$q$, $\nu$, and $\epsilon$. In this remark, we denote by $C(X,N,S_0,D_0,q)$ this function when  $\nu$ and $\epsilon$ are given. From~\eqref{dynS}, \eqref{dynD} and~\eqref{CostX}, we have the straightforward property
\begin{equation}\label{CostX_symm}
C(X,N,S_0,D_0,q)=C(-X,-N,-S_0,-D_0,q). 
\end{equation}
Observing  that
$qC(X)=\int_{[0,T)} qP_u  \textup{d}X_u
 +  \frac{1}{2} \underset{0\leq\tau<T}{\sum} (\Delta X_\tau )^2   -  (qP_T) X_T  +   \frac 12 (X_T)^2$, and remarking that
 $qS$ and $qD$ satisfy~\eqref{dynS} and~\eqref{dynD} with $q=1$, we also get
\begin{equation}\label{CostX_symm2}
qC(X,N,S_0,D_0,q)=C(X,N,qS_0,qD_0,1). 
\end{equation}
\end{remark}
\begin{remark}\label{notations_ladlag}
Since $X$ is a c\`agl\`ad process and $N$ is a c\`adl\`ag process, we will have to work with  l\`adl\`ag (with finite right-hand and left-hand limits) processes. When $Z$ is a l\`adl\`ag process, we set
$\Delta^- Z_t = Z_{t} - Z_{t-}$ and $\Delta^+ Z_t = Z_{t^+} - Z_t$ the left and right jumps of~$Z$, and $Z^\text{c}_t = Z_t - \sum_{0 \leq \tau < t} \Delta^+ Z_\tau - \sum_{0 < \tau \leq t} \Delta^- Z_\tau $ the continuous part of $Z$. We also set $\Delta Z_t = Z_{t+} - Z_{t-}$ and use the shorthand notation $\dd Z_t=\dd Z^c_t+\Delta Z_t$. If $\dd Z_t=\dd \tilde{Z}_t$ for some other l\`adl\`ag process~$\tilde{Z}$, this means that $\dd Z^c_t=\dd \tilde{Z}^c_t$ and $\Delta Z_t=\Delta \tilde{Z}_t$. In particular, when $Z$ is c\`adl\`ag and $\tilde{Z}$ is c\`agl\`ad, this means that  $Z_{t} - Z_{t-}=\tilde{Z}_{t^+} - \tilde{Z}_t$ at the jump times. 
\end{remark}

Then, the optimal execution problem consists in finding an admissible strategy~$X$ that minimizes the expected cost~$\E[C(X)]$ for a given initial position~$x_0\in \R$. This problem for $x_0=0$ is directly related to the existence of Price Manipulation Strategies as defined below.
\begin{definition}\label{def:PMS}
A Price Manipulation Strategy (PMS) in the sense of Huberman and Stanzl~\cite{HS} is an admissible strategy $X$ such that  $X_0 = X_{T^+} = 0$ a.s. for some $T>0$ and $\E[C(X)]<0$.
\end{definition}

We have the following result that gives a necessary and sufficient condition to exclude PMS.
\begin{theorem}\label{Thm_Pmg} The model does not admit PMS if, and only if the process $P$ is a $(\mathcal{F}_t)$-martingale when $X \equiv 0$. In this case, the optimal strategy $X^\textup{OW}$ is the same as in the Obizhaeva and Wang~\cite{OW} model. It is given by
\begin{align}\label{opt_strat_Pmg}
 \Delta X^\textup{OW}_0  =  -  \frac{ x_0}{2+\rho T}, \quad 
 \Delta X^\textup{OW}_T  =  - \frac {x_0} {2+\rho T} , \quad 
 \textup{d} X^\textup{OW}_t  = -  \rho \frac { x_0} {2 + \rho T} \ \textup{d}t \ \text{ for } t\in (0,T),
\end{align}
and has the expected cost  $\mathbb{E}[C(X^\textup{OW})]= - P_0  x_0  + \left[ \frac {1-\epsilon} {2 + \rho (T-t)}  +  \frac\epsilon2 \right] x_0^2/q$.
\end{theorem}
This theorem is proved in Appendix~\ref{appendix:proof_Thm_Pmg}.
Similar results are standard in financial mathematics, but to the best of our knowledge, it has not yet been formulated as such in the literature in a context with price impact and with respect to the notion of Price Manipulation Strategies. In usual optimal execution frameworks, the unaffected price is assumed \textit{a priori} to be a martingale, which is not the case here.
Note that if $P$ is a martingale, the optimal strategy is very robust in the sense that it does not depend on~$N$, and is therefore the same as the one in the Obizhaeva and Wang model~\cite{OW} that corresponds to $N\equiv 0$ and $D_0=0$. In fact, it does not depend either on~$\epsilon$ and $\nu$, and only depends on~$\rho$.

Theorem~\ref{Thm_Pmg} indicates that suitable models for the order flow~$N$ should be such that $P$ is, roughly speaking, close to a martingale when the strategic trader is absent, so that arbitrage opportunities are short-lived and not too visible.
 This raises at least three questions. Which ``simple'' processes $N$ can lead to a martingale price~$P$? Can we characterize the optimal strategy when $P$ is not a martingale? In particular, in the latter case, how does the optimal strategy take the market orders issued by other participants into account?
 In this paper, we study these questions when $N$ follows a Hawkes process.

\begin{remark}\label{model_plus_mg}
The model can be generalized by adding a c\`adl\`ag $(\cF_t)$-martingale $S^0$ to the price process $P$, i.e. if we replace~\eqref{defP} by
$P_t = S_t + D_t + S^0_t$, with $S^0_0 = 0$. This does not change the optimal execution problem since, using an integration by parts, $S^0$ adds the following term to the cost 
\begin{eqnarray}
\int_{[0,T)} S^0_t \ \textup{d}X_t \ - \ S_T^0 X_T
& = & S_T^0 X_T \ - \ S_0^0 X_0 \ - \ \int_{[0,T)} X_t \ \textup{d}S^0_t 
\ - \ S_T^0 X_T
\nonumber \\
& = & - \ \int_{[0,T)} X_t \ \textup{d}S^0_t,
\nonumber
\end{eqnarray}
which has a zero expected value from the martingale property. Let us note that there is no covariation between the processes~$X$ and $S_0$ since they do not jump simultaneously and $X$ has bounded variations. 
\end{remark}
\begin{remark}\label{Rk_N_mg}
Similarly, when $N$ is a c\`adl\`ag $(\cF_t)$-martingale and $X$ is an admissible liquidating strategy for $X_0=x_0$, we have 
$$\E[C(X)]=\E \left[\int_{[0,T)} D_u \ \textup{d}X_u 
\ + \ \frac{1-\epsilon}{2q} \underset{0\leq\tau<T}{\sum} (\Delta X_\tau)^2
 \ - \ D_T X_T \ + \ \frac{1-\epsilon}{2q} \ X_T^2 \right] +  \frac \epsilon {2q} x_0^2, $$
since $x_0^2=\int_{[0,T+]}\dd [(X_t-X_0)^2]=2\int_{[0,T)}(X_u-X_0) \dd X_u +   \underset{0\leq\tau<T}{\sum} (\Delta X_\tau )^2 - 2(X_T-X_0)X_T +X_T^2$. When $\epsilon\in [0,1)$, we set $X^\epsilon_t=(1-\epsilon)X_t$ and get
\begin{equation}\label{cout_N_mg}
\E[C(X)]=\frac{1}{q(1-\epsilon)}\E \left[\int_{[0,T)} qD_u  \textup{d}(X^\epsilon_u) 
 +  \frac{1}{2} \underset{0\leq\tau<T}{\sum} (\Delta X^\epsilon_\tau)^2
  -  qD_T X^\epsilon_T  +  \frac{1}{2}  (X^\epsilon_T)^2 \right] +  \frac \epsilon {2q} x_0^2.
\end{equation}
Therefore, $X$ is optimal if, and only if $X^\epsilon$ is optimal in the model with $\epsilon = \nu = 0$, $q=1$ and an incoming flow of market orders equal to $(1-\nu)N$. 
\end{remark}

\subsection{The MIH model}\label{section:MIH}

\subsubsection{Definitions and notations}
We introduce the MIH (Mixed-market-Impact Hawkes) price model, where
\begin{equation}
N_t = N_t^+ - N_t^-,
\nonumber
\end{equation}
the process $(N^+, N^-)$ being a symmetric two-dimensional marked Hawkes process 
of intensity $(\kappa^+, \kappa^-)$. The process  $(N^+,N^-,\kappa^+,\kappa^-)$ is c\`adl\`ag and jumps when $N$ jumps.
We note $n^+(\textup{d}t, \textup{d}v)$ and $n^-(\textup{d}t, \textup{d}v)$ the Poisson measures on $\mathbb{R}^+ \times \mathbb{R}^+$ associated to $N^+$ and $N^-$ respectively, where the variable $v$ stands for the amplitudes of the jumps, i.e. the volumes of incoming market orders. We restrain to the case of i.i.d. unpredictable marks of common law $\mu$ on $\mathbb{R}^+$, i.e. for any $A \in \mathcal{B}(\mathbb{R}^+)$ and
$t \geq 0$,
\begin{equation}
\kappa^\pm_t \ \mu(A) = \lim_{h \rightarrow 0^+} \frac1h \E[n^\pm([t,t+h],A) | \mathcal{F}_t ],
\nonumber
\end{equation}
where $\mathcal{F}_t = \sigma\left(N^+_u, N^-_u, u \leq t\right) = \sigma\left(N_u, u \leq t\right)$ as defined earlier. In other words, at time $t$, the  conditional instantaneous jump intensity of $N^\pm$ is given by $\kappa^\pm_t$, and the amplitudes of the jumps are i.i.d. variables of law $\mu$ which are independent from the past, i.e. from $\mathcal{F}_{t^-}$. We also define
$$m_k = \int_{\mathbb{R}^+} v^k  \mu(\textup{d}v), \ k \in \N,$$
assuming moreover that $m_2<\infty$.
We choose the Hawkes kernel to be the exponential $t \mapsto \exp(-\beta t), \ \beta \geq 0$, so that $(N^+,N^-,\kappa^+,\kappa^-)$ is Markovian. 
Thus, we set
\small
\begin{equation}
\begin{pmatrix}
\kappa^+_t \\
\kappa^-_t
\end{pmatrix}
=
\begin{pmatrix}
\kappa_\infty \\
\kappa_\infty
\end{pmatrix}
 + \left[
\begin{pmatrix}
\kappa^+_0 \\
\kappa^-_0
\end{pmatrix}
-
\begin{pmatrix}
\kappa_\infty \\
\kappa_\infty
\end{pmatrix}
\right] \exp(-\beta t)
+ \int_0^t \int_{\left(\mathbb{R}^+\right)^2}
\exp(-\beta (t-u))
\begin{pmatrix}
\phis(v^+/m_1) & \phic(v^-/m_1) \\
\phic(v^+/m_1) & \phis(v^-/m_1) 
\end{pmatrix}
.
\begin{pmatrix}
n^+(\textup{d}u, \textup{d}v^+) \\
n^-(\textup{d}u, \textup{d}v^-) 
\end{pmatrix},
\nonumber
\end{equation}
\normalsize
where $\kappa_\infty \geq 0$ is the  common baseline intensity of $N^+$ and $N^-$, and $\phis, \phic: \mathbb{R}^+ \rightarrow \mathbb{R}^+$ are measurable positive functions that satisfy 
\begin{equation}
\ios := \int_{\mathbb{R}^+} \phis(v/m_1)  \mu(\textup{d}v) <\infty
\quad , \quad
\ioc := \int_{\mathbb{R}^+} \phic(v/m_1)  \mu(\textup{d}v) <\infty.
\nonumber
\end{equation}
We assume besides that
$$
\int_{\mathbb{R}^+} \phis^2(v/m_1)  \mu(\textup{d}v) <\infty, \
\int_{\mathbb{R}^+} \phic^2(v/m_1)  \mu(\textup{d}v) <\infty
$$
to have $\sup_{s \in [0,t]} \E[N_s^2]< \infty$, and we note that this property is automatically satisfied when $\phis$ and $\phic$ have a sublinear growth since we have assumed $m_2<\infty$. From the modeling point ov view, we may expect that the functions $\phis$ and $\phic$ are nondecreasing: the larger an order is, the more other orders it should trigger. However, we do not need this monotonicity assumption in the mathematical analysis.

%
%
%
Equivalently, in this Markovian setting, the intensities $\kappa^+_t$ and $\kappa^-_t$ follow the dynamics
\begin{align}
\textup{d}\kappa^+_t  &= \ - \beta \ (\kappa^+_t - \kappa_{\infty}) \ \textup{d}t
\ + \ \phis(\textup{d}N^+_t/m_1) + \ \phic(\textup{d}N^-_t/m_1),
\nonumber \\
\textup{d}\kappa^-_t  &= \ - \beta \ (\kappa^-_t - \kappa_{\infty}) \ \textup{d}t
\ + \ \phic(\textup{d}N^+_t/m_1) + \ \phis(\textup{d}N^-_t/m_1),
\label{dyn_kappas}
\end{align}
where formally,
$\int_0^t \phis(\textup{d}N^+_u/m_1) = 
\int_0^t \int_{\mathbb{R}^+} \phis(v/m_1) \ n^+(\textup{d}u, \textup{d}v)$ for $t \geq 0$.
%
As pointed out in Hardiman, Bercot and Bouchaud~\cite{HBB} and Bacry and Muzy~\cite{BacryMuzy} for instance (in a slightly different context since in our framework, $N$ models market orders only), a power-law Hawkes kernel is more in accordance with market data than an exponential one. It is possible in principle to approximate a completely monotone decaying kernel with a multi-exponential one while preserving a Markovian framework, at the cost of increasing the dimension of the state space, see for example Alfonsi and Schied~\cite{AS_SICON}. This investigation is left for future research. 

Note that $N^+$ and $N^-$ boil down to independent composed Poisson processes in the case $\beta = 0, \ \phis = \phic \equiv 0$. 
The meaning of the parameters is rather clear: $\kappa^+$ and $\kappa^-$ are mean reverting processes, and $\iota_\textup{s}$ and $\iota_\textup{c}$ respectively describe how a market buy  order increases the instantaneous probability of buy (resp. sell) orders. More precisely, $\iota_\textup{s}$ encodes both the splitting of meta-orders, and the fact that participants tend to follow market trends (which is called the herding effect).
On the other hand, $\iota_\textup{c}$ describes opportunistic traders that sell (resp. buy) after a sudden rise (resp. fall) of the price. The functions $\phis$ and $\phic$ allow respectively the self and cross-excitations in the order flow to depend on the volumes of the orders. For instance, for constant functions $\phis \equiv \ios$ and $\phic \equiv  \ioc$, the model boils down to the standard Hawkes model where 
$\kappa^\pm$ makes jumps of constant size when $N^\pm$ jumps.

Hawkes processes have been recently used in the literature to model the price. In particular, Bacry et al.~\cite{BDHM} consider a similar model where $N$ models all price moves, with $\nu=1$, $\iota_\textup{s}=0$ and deterministic jumps (i.e. $\mu$ is a Dirac mass). More recently, Bacry and Muzy~\cite{BacryMuzy} have proposed an four-dimensional Hawkes process to model the market buy and sell orders together with the up and down events on the price. In contrast, the model that we study here determines the price impact of an order in function of its size.  For the reader who is not accustomed to Hawkes processes, we point the original paper~\cite{HawkesMarks}, the paper by Embrechts et al.~\cite{ELL} for an overview of multivariate marked Hawkes processes and the book of Daley and Vere-Jones~\cite{DVJ} for a more detailed account.

\begin{remark}
As one can see in Equation~\eqref{dyn_kappas}, the orders of the strategic trader do not impact the jump rates $\kappa^+$ and $\kappa^-$ (there is no $\textup{d}X_t$ term), as opposed to the market orders issued by other traders. The first reason for this modeling choice is tractability. However, it is found empirically by T\`oth et al.~\cite{TPLF} that the main contribution to the self-excitation of the order flow comes from the splitting effect. Each individual trader tends to post several orders of the same nature (buy or sell) in a row, which creates auto-correlation in the signs of trades, and this effect is significantly stronger than the mutual excitation between different traders. Thus, it is an acceptable approximation to neglect the excitation coming from the orders of the strategic trader. Of course, it would be nice to find in the future a tractable model that gives a unified framework for the mutual excitation that considers equally all the market orders. 
\end{remark}

\subsubsection{Stationarity and low-frequency asymptotics of the MIH model}\label{section:low-freq_MIH}
Up to now, we have presented the MIH model without assuming stationarity. In most models featuring Hawkes processes, stationarity is an \textit{a priori} assumption, but here, we do not need it to derive the optimal strategy. However, if one wishes to use the MIH model with constant parameters on a large time period, it may be reasonable to consider parameters that satisfy stationarity. This is why we present here a few results that are standard in the literature of Hawkes processes. 

We consider the MIH model when the strategic trader is absent, i.e. $X\equiv 0$.
\begin{proposition}\label{prop_kappa_stat}
The process $(\kappa^+_t,\kappa^-_t)$ converges to a stationary law if, and only if $\ios + \ioc <\beta$.
\end{proposition}
\begin{proof}
We can apply the results of the existing literature on marked Hawkes processes with unpredictable marks (for instance Hawkes and Oakes~\cite{HawkesOakes}, Br\'emaud ans Massouli\'e~\cite{BremaudMassoulie} or Daley and Vere-Jones~\cite{DVJ}) to obtain that $(\kappa^+_t,\kappa^-_t)$ converges to a stationary law if the largest eigenvalue of
\begin{equation}
\int_{\mathbb{R}^+ \times \mathbb{R}^+} \exp(-\beta t) \
\begin{pmatrix}
\phis(v/m_1) & \phic(v/m_1) \\
\phic(v/m_1) & \phis(v/m_1)
\end{pmatrix}
\ \textup{d}t \ \mu(\textup{d}v) =  \frac{1}{\beta} \begin{pmatrix}
\ios & \ioc \\
\ioc & \ios
\end{pmatrix}
\nonumber
\end{equation}
is strictly below unity. Conversely, if $\ios + \ioc \ge \beta$, we have $$\frac{\dd}{\dd t} 
\E[\kappa^+_t + \kappa^-_t]=2\beta \kappa_\infty + (\ios + \ioc - \beta)\E[\kappa^+_t + \kappa^-_t] \ge 2\beta \kappa_\infty$$
and the process cannot be stationary. 
\end{proof}

We now  study the low-frequency asymptotics of the price process $P$ in the MIH model. We consider the sequence $P^{(n)}_t = P_{nt}/\sqrt{n}$ for $n \geq 1$. We have
$P^{(n)}_t = S^{(n)}_t + D^{(n)}_t$, where we also set $S^{(n)}_t = S_{nt}/\sqrt{n}$ and $D^{(n)}_t = D_{nt}/\sqrt{n}$. To study the behaviour of $D^{(n)}$, we need the following lemma. 
\begin{lemma}\label{lem:convegence_espDt2}
When $\ios+\ioc<\beta$, the expectation $\E[D_t^2]$ converges to a finite positive value as $t\rightarrow +\infty$. 
\end{lemma}
The proof of this lemma is rather straightforward. We just have to calculate $\E[\delta_t^2]$, $\E[\delta_t D_t]$ and $\E[D_t^2]$ and check that these expectations converge when  $\ios+\ioc<\beta$.
This result implies that $(D^{(n)}_{t_1},\dots,D^{(n)}_{t_k})$ converges to zero for the $L^2$ norm for any $0\le t_1\le \dots \le t_k$. This gives that the process $D^{(n)}$ converges to zero.

We thus focus on the convergence of $S^{(n)}_t= \frac{\nu}{q} \frac{N^+_{nt}-N^-_{nt}}{\sqrt{n}}$. If the jumps of $N$ are bounded, i.e. $\mu$ has bounded support, and $v \mapsto \phis(v/m_1)$ and $v \mapsto \phic(v/m_1)$ are bounded on the support of $\mu$ (which are reasonable assumptions in practice), a straightforward adaptation of Corollary~1 of Bacry et al.~\cite{BDHM2} gives the convergence in law of $S^{(n)}$ to a non-standard Brownian motion with zero drift.

\section{Main results}\label{section:mr}

Now that the whole framework is set up, we present the main results of the present paper.
 
\begin{itemize}
\item The optimal execution problem can be solved explicitly in the MIH model and the optimal strategy has still a quite simple form, see Theorem~\ref{theo:opt_strat_Hawkes}. Of course, this result relies on the assumptions of linear price impact and exponential decay kernel, which are not in accordance with empirical facts, see for example Potters and Bouchaud~\cite{PB} and Bouchaud et al.~\cite{BGPW}. We mention here that it would be possible to keep an affine structure of the optimal strategy by considering complete monotone decay kernels as in Alfonsi and Schied~\cite{AS_SICON}. However, we believe that the optimal strategy is interesting at least from a qualitative point of view, since it gives clear insights on how to react optimally to observed market orders and on the role of the different parameters of the model.

\item Price Manipulation Strategies necessarily appear when the flow of market orders is Poissonian, and they are rather robust in the sense that they can be implemented without knowing the model parameters. Namely, the strategy which consists in trading instantly a small proportion of the volume of each incoming market order in the opposite direction is profitable on average, see Proposition~\ref{prop:poisson:arbitrage}.
This justifies to consider more elaborate dynamics for the order arrivals.

\item Even in a non-Poissonian MIH setup, Price Manipulation Strategies can arise. Depending on the parameters of the model and on the size of each observed market order, one should either trade instantly in the opposite direction to take market resilience into account, or in the same direction to take advantage of the self-excitation property of Hawkes processes. However, our framework allows for a specific equilibrium to take place, that we call the Mixed-market-Impact Hawkes Martingale (MIHM) model, where PMS disappear.

\item In the MIHM model, one has in particular $\iota_\textup{s} > \iota_\textup{c}, \ \nu < 1$ and $\beta = \rho$, and the self-excitation property of the order flow exactly compensates the price resilience induced by market makers. The resulting price process is a martingale even at high frequencies, and in this case we find that the optimal strategy and cost function are those of  Obizhaeva and Wang~\cite{OW}. The conditions of this model imply that if $\iota_\textup{c} = 0$, the norm $\ios/\beta$ of the Hawkes kernel that symbolizes the endogeneity ratio of the market, see Filimonov and Sornette~\cite{FilimonovSornette}, should be equal to $1-\nu$, i.e. the proportion of market impact which is transient.

\item The fact of reacting to the market orders of other traders with instantaneous market orders can trigger chain reactions and lead to market instability. We show that in the MIH framework, the conditions under which it is profitable for the strategic trader to react instantaneously to other trades are quite equivalent to the existence of PMS. Although the model is clearly a simplified view of the market, it is remarkable to obtain in this case such a clear connection between market stability and free profits. It would be interesting for further reasearch to investigate if this conclusion still holds in a more general model.

\end{itemize}

\section{The optimal strategy}\label{section:opt_strat}

We need to introduce some notations to present the main results on the optimal execution.  Instead of working with $\kappa^+_t$ and $\kappa^-_t$, we will rather use $\delta_t = \kappa^+_t - \kappa^-_t$ and $\Sigma_t = \kappa^+_t + \kappa^-_t$ that satisfy from~\eqref{dyn_kappas}
%
%
\begin{equation}
\textup{d}\delta_t  \ = \ - \beta \ \delta_t \ \textup{d}t
\ + \ \text{d}I_t
\quad , \quad
\textup{d}\Sigma_t  \ = \ - \beta \ (\Sigma_t - 2 \kappa_{\infty}) \ \textup{d}t
\ + \ \text{d}\overline{I}_t,
\label{dyn_deltasigma}
\end{equation}
%
where
\begin{align}
I_t &= \int_0^t \left[(\phis-\phic)(\textup{d}N^+_u/m_1) - (\phis-\phic)(\textup{d}N^-_u/m_1)\right],
\nonumber \\
\overline{I}_t &= \int_0^t \left[(\phis+\phic)(\textup{d}N^+_u/m_1) + (\phis+\phic)(\textup{d}N^-_u/m_1)\right].
\label{def_II}
\end{align}
The processes $I$ and $\overline{I}$ are c\`adl\`ag processes which describe intensity jumps, and their jump times are those of $N$.
In the standard Hawkes framework where $\phis$ and $\phic$ are constant, one has $\phis \equiv \ios$ and $\phic \equiv \ioc$, and when $N$ jumps,
$I$ jumps of $(\ios-\ioc) \ \sgn(\Delta N_t)$  and $\overline{I}$ of $\ios+\ioc$.

We note $(\tau_i)_{i \geq 1}$ the ordered random jump times of $N$ and set $\tau_0 = 0$.
For $t \in [0,T]$, we also note $\chi_t$ the total number of jumps of $I$ that occurred between time $0$ and time $t$. From~\eqref{dyn_deltasigma}, we have
%
%
\begin{equation}
\delta_t \ = \ \delta_0 \exp(-\beta t) 
\ + \ \overset{\chi_t}{\underset{l=1}{\sum}} \exp(-\beta(t-\tau_l)) \Delta I_{\tau_l}
\ = \ \delta_0 \ \exp(-\beta t) + \exp(-\beta t) \ \Theta_{\chi_t},
\nonumber
\end{equation}
where we define $\Theta_0 = 0$ and 
%
%
\begin{equation}
\Theta_i \ = \ \sum_{l=1}^i \exp(\beta \tau_l) \Delta I_{\tau_l}
\ = \ \underset{0 < \tau \leq \tau_i}{\sum} \exp(\beta \tau) \ \Delta I_\tau, \quad  i \ge 1.
\nonumber
\end{equation}
For $i \geq 0$ and $t \in [\tau_i,\tau_{i+1})$, we obtain that 
$\delta_t \exp(\beta t) = \delta_0 + \Theta_i$
 only depends on $t$ through the integer $i = \chi_t$. 
We introduce the useful quantities
\begin{equation}
\alpha = \ios -\ioc, \quad \eta = \beta-\alpha,
\nonumber
\end{equation}
and the two continuously differentiable functions $\zeta, \omega : \mathbb{R} \rightarrow \mathbb{R}^+$ defined by
\begin{align}
\zeta(0) & = 1 \quad \text{ and } \quad \forall y \neq 0, \quad \zeta(y) \ = \ \frac{1-\exp(-y)}y,
 \label{def_zeta} \\
\zeta'(0) & = -1/2 \quad \text{ and } \quad \forall y \neq 0, \quad \zeta'(y) \ = \ \frac{(1+y)\exp(-y) - 1}{y^2}
\ = \ \frac{\exp(-y)-\zeta(y)}y,
\nonumber \\
\omega(0) & = 1/2 \quad \text{ and } \quad \forall y \neq 0, \quad 
\omega(y) \ = \ \frac{\exp(-y) - 1 + y}{y^2}
 \ = \ \frac{1-\zeta(y)}y,
\label{eqn:def_omega} \\
\omega'(0) & = -1/6 \quad \text{ and } \quad \forall y \neq 0, \quad \omega'(y) \ = \ \frac{2(1-\exp(-y))-y(1+\exp(-y))}{y^3}
\ = \ \frac{2 \zeta(y) - 1 - \exp(-y)}{y^2}.
\nonumber
\end{align}
Both functions non-increasing, diverge to $+\infty$ at negative infinity and vanish at positive infinity. Let us now enounce the main theorem for the optimal execution problem.

%



\begin{theorem}\label{theo:opt_strat_Hawkes}
Let $\epsilon\in[0,1)$. The optimal strategy $X^*$ that minimizes the expected cost $\E[C(X)]$ among admissible strategies that liquidate $x_0$ assets is explicit. It is a linear combination of $(x_0,D_0,\delta_0, I, N)$ and can be written as
\begin{equation}
X^* = X^\textup{OW} + X^\textup{trend} + X^\textup{dyn},
\nonumber
\end{equation}
where 
\begin{itemize}
\item $X^\textup{OW}$ is the optimal strategy in the Obizhaeva and Wang~\cite{OW} model, given by~\eqref{opt_strat_Pmg} in Theorem~\ref{Thm_Pmg},
\item $X^\textup{trend}$ is the \enquote{trend strategy}, given by~\eqref{Xtrend_eta}.
\item $X^\textup{dyn}$ is the \enquote{dynamic strategy}, given by~\eqref{Xdyn_eta}.
\end{itemize}
The strategy $X^\textup{OW}$ is a linear function of~$x_0$, $X^\textup{trend}$ is a linear function of~$(D_0,\delta_0)$ while $X^\textup{dyn}$ is a linear function of the processes $I$ and $N$. The discontinuity times of $X^\textup{dyn}$ are those of $N$, and if $N$ jumps at time $\tau \in (0,T)$, we have
\begin{equation}
(1-\epsilon) \Delta X^\textup{dyn}_\tau =
\  \frac{1 + \rho (T-\tau)}{2 + \rho (T-\tau)} \left\{\frac {m_1} \rho \ \Delta I_\tau - (1-\nu) \ \Delta N_\tau \right\}
+ \frac{m_1}{2\rho} (\nu \rho-\eta) \frac{\rho (T-\tau)^2 \times \omega(\eta(T-\tau))}{2+\rho (T-\tau)} \ \Delta I_\tau.
\label{reaction_trade}
\end{equation}
%
%
%
All explicit formulas are given in Appendix~\ref{appendix:formulas_opt_strat}. The value function of the problem is given by
\begin{eqnarray}
q \times \mathcal{C}(t,x,d,z,\delta,\Sigma) & = &
- q (z+d) x
\ + \ \left[ \frac{1-\epsilon}{2+\rho(T-t)} + \frac\epsilon 2 \right] x^2
\ + \ \frac{\rho (T-t)}{2+\rho(T-t)} \left[ qd - \ \mathcal{G}_\eta(T-t) \ \frac{\delta m_1}\rho \right] x
\nonumber \\
& & \nonumber \\
& & \
- \ \frac1{1-\epsilon} \times \frac{\rho (T-t)/2}{2+\rho(T-t)} 
\left[ qd - \ \mathcal{G}_\eta(T-t) \ \frac{\delta m_1}\rho \right]^2
\ + \ \hat c_\eta (T-t) \ \left(\frac{\delta m_1}\rho\right)^2 \nonumber \\
& & \nonumber \\
& & \ + \ e(T-t) \ \Sigma \ + \ g(T-t),
\nonumber
\end{eqnarray}
where for $u \in [0,T]$,
\begin{eqnarray}
\mathcal{G}_\eta(u) & = & \zeta(\eta u)+ \nu \rho u \ \omega(\eta u),
\nonumber \\
\hat c_\eta(u) & = &  \frac1{1-\epsilon} \times (\eta- \nu \rho)^2 \frac{\rho u^3}{8} \omega'(\eta u) \zeta(\eta u).
\nonumber
\end{eqnarray}

The functions $e$ and $g$ are the unique solution of the differential equations~\eqref{eqn:e_general} and~\eqref{eqn:g_general} with $e(0)=g(0)=0$.
\end{theorem}
 The proof of this theorem is given in Appendix \ref{appendix:proof_opt_ctrl}. Let us mention here that the functions~$e$ and $g$ admit explicit forms by the mean of the exponential integral function, that are very cumbersome. They can be obtained by using a formal calculus software such as Mathematica. Since they do not play any role to determine the optimal strategy and require several pages to be displayed, we do not give these explicit formulas. Note that they are simpler in the case $\eta=0$, for which the explicit formulas are given by Equations~\eqref{val_e_crit} and~\eqref{val_g_crit}.

The optimal strategy~$X^*$ is illustrated on Figure~\ref{fig:Hawkes_opt_strat} for two different sets of parameters. It is worth to notice that the strategy is linear with respect to $x_0$, $D_0$, $\delta_0$, $I$ and~$N$. This property is due to the affine structure of the model and the quadratic costs. In particular, the reaction of the optimal strategy to the other trades does not depend on~$x_0$.  The strategy $X^\textup{trend}$ is the part of the strategy which is proportional to $D_0$ and $\delta_0$ and thus takes advantage of temporary price trends that are known at time $0$. The strategy $X^\textup{dyn}$ is proportional to the processes $I$ and $N$ and describes the optimal reactions to observed price jumps. Last, let us stress that having an explicit formula for the optimal strategy is an important feature to use it in practice. Since the strategy reacts to each market order (or at least to those which trigger price moves), its computation time should be significantly lower than the typical duration between two of these orders.


%
\begin{figure}[H]
	\centering
\subfigure[$\rho=25$]{
    \includegraphics[height=6cm,width=8cm]{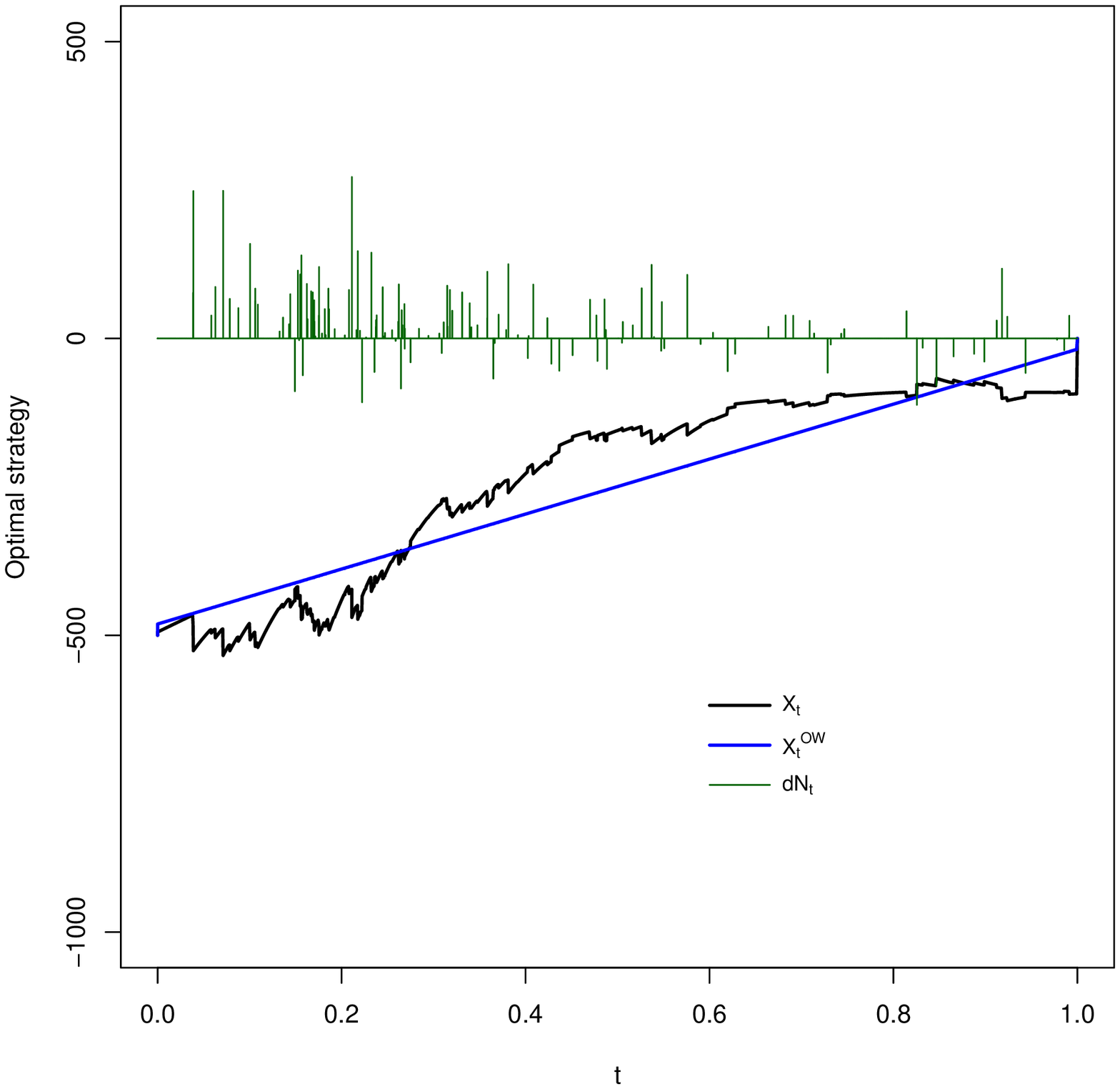}
	\label{fig:Xt_rho}
    }
\subfigure[$\rho=16$]{
    \includegraphics[height=6cm,width=8cm]{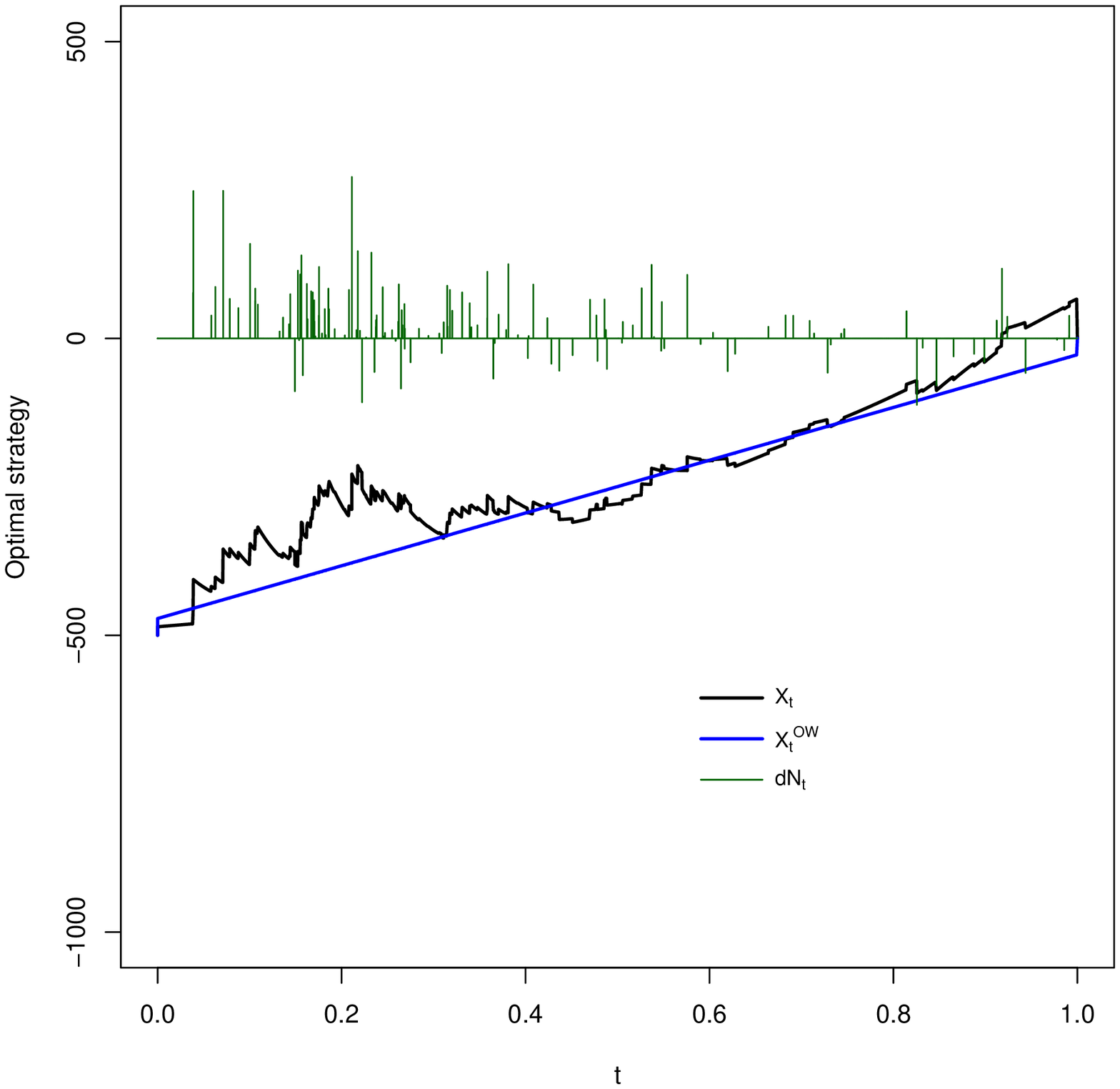}
	\label{fig:Xt_beta} 
    }
	\caption{Optimal strategy in the Hawkes model, in black, for $q=100, \ T=1, \ \beta = 20, \ \ios=16, \ \ioc=2, \ \kappa_\infty = 12, \ \epsilon = 0.3, \ \nu = 0.3, \ D_0=0.1, \ \kappa^+_0= \kappa^-_0=60, \ m_1 = 50, \ X_0 = -500, \ \mu=\text{Exp}(1/m_1), \
\phis(y) = 1.2 \times y^{0.2} + 0.5 \times y^{0.7} + 14.4 \times y, \
\phic(y) = 1.2 \times y^{0.2} + 0.5 \times y^{0.7} + 0.4 \times y$ for all $y>0$.
The strategy of the Obizhaeva and Wang model is given in blue as a benchmark, and the jumps of $(N_t)$ are plotted in green (with the same trajectory for the two graphs). On the left graph, $\ios < \beta<\rho$ and the strategy is based on mean-reversion: each time $N$ jumps, $X$ jumps in the opposite direction. On the right graph, $\rho=\ios < \beta$ and the strategy is trend-following.}
    \label{fig:Hawkes_opt_strat}
\end{figure}

Let us make some comments on the optimal strategy, and more precisely on how the strategic trader reacts to the orders issued by other traders. First, we observe from~\eqref{Xdyn_eta} that the block trades that immediately follow jumps of $N$ are then compensated by the continuous trading rate. When $\phis=\phic$, we have $I\equiv 0$ and these block trades, as given by~\eqref{reaction_trade}, are always opposed in sign to the market orders that they follow. For general functions $\phis$ and $\phic$, the signs of these trades depend on the size of the last preceding jumps of $N$. 
For example, in the case where $\eta= \nu \rho$, the strategic trader makes a trade in the opposite direction if 
$|\textup{d}N_t|> \frac{m_1}{\rho(1-\nu)} (\phis-\phic)(|\textup{d}N_t|/m_1)$, but trades in the same way otherwise.
The same conclusion holds for any parameter value when $T-t \rightarrow 0$ since $\rho (T-t)^2 \times \omega(\eta(T-t))$ vanishes. We now consider the asymptotics when the trading horizon is large: in this case, it is reasonable to assume that $\eta>0$ which is required to get stationary intensities $\kappa^+$ and $\kappa^-$, see Section~\ref{section:low-freq_MIH}. Then, when $T-t  \rightarrow +\infty$, the jump part of $X^\text{dyn}$ given by~\eqref{reaction_trade} can be well approximated by
%
%
$$ \frac{m_1}{2\rho} \left(1+ \frac{\nu \rho}{\eta}\right) \textup{d}I_t- (1-\nu) \ \textup{d} N_t.$$
%
Therefore, the strategic trader makes a trade in the opposite direction if 
$|\textup{d}N_t|> \frac{m_1}{2\rho(1-\nu)}(1+ \frac{\nu \rho}{\eta}) (\phis-\phic)(|\textup{d}N_t|/m_1)$
 and trades in the same direction otherwise. In the case $\iota_\textup{c}=0$ and $\phis \equiv \ios$ where there is only volume-independent self-excitation, we can interpret this behavior as follows: if a market buy order is relatively small, it may be a part of a big split order, and thus be followed by other buy orders that will make the price go up, and the strategic trader has interest to follow this trend. However, if a market buy order is relatively big, the price resilience effect is likely to dominate and the strategic trader has interest to trade in the opposite way. 

Last, it is interesting to notice that  the optimal strategy  only depends on $(\phis,\phic)$ through $\phis-\phic$. This key self-excitation function tunes the way that the strategic trader should react to other market orders.

\begin{remark}\label{rk:Poisson_strat}
The MIH model with $\eta=0$ includes the particular case of independent Poisson processes when $\beta = 0$ and $\phis=\phic \equiv 0$. In that case, if $N$ jumps at time $\tau \in (0,T)$, we get from~\eqref{reaction_trade}
\begin{equation}
(1-\epsilon) \Delta X^\textup{dyn}_\tau \ = \
\ - \ \frac{1 + \rho (T-\tau)}{2 + \rho (T-\tau)} \times (1-\nu) \ \Delta N_\tau.
\nonumber
\end{equation}
Since the self-excitation effect is removed, the price is a mean-reverting process when the strategic trader is passive. Thus, each time a market order is observed, the optimal strategy consists in posting immediately a market order in the opposite direction, to arbitrage the resilience of the price. Such an obvious Price Manipulation Strategy is unrealistic, therefore modeling the order flow with Poisson processes is not satisfactory. We refer to Section~\ref{section:Poisson_arbitrages} for more details.
\end{remark}

\begin{remark}\label{rk:}
Following Remark~\ref{rk:add_liq_cost}, a natural question is to look at the quantity~$x_0$ that minimizes $\E[C(X)]+P_0\times x_0$, i.e. the expected liquidation cost with respect to the mark-to-market value. From Theorem~\ref{theo:opt_strat_Hawkes} we obtain easily that, at time~$0$, this quantity is minimal for 
$$ x_0^*= \frac{ \rho T [q D_0 - \mathcal{G}_\eta(T)\frac{\delta_0 m_1}{\rho} ]}{2\left( 1+ \frac{\epsilon}{2} \rho T \right)}.$$
We can give a simple heuristic for the sign of $x_0^*$: when $D_0 \ge 0$ and $\delta_0 \le 0$ the price trend is negative and it is more favorable to sell ($x_0\ge 0$) since $\mathcal{G}_\eta$ is nonnegative. 
\end{remark}

\section{Price Manipulation Strategies in the MIH model}\label{section:PMS}

In this section, we study Price Manipulation Strategies (PMS), as introduced by Definition~\ref{def:PMS}, in the context of the MIH model. As a matter of fact, the value function given in Theorem~\ref{theo:opt_strat_Hawkes} can be negative even for $x_0 = 0$, which would constitute a PMS. We first determine necessary and sufficient conditions on the parameters of the model to exclude such strategies. Then, we study the particular case of Poisson processes, which may seem natural to model the order flow but allow for robust arbitrages to arise in this framework.
\subsection{The Mixed-market-Impact Hawkes Martingale (MIHM) model}\label{section:MIHM_model}
Theorem~\ref{Thm_Pmg} gives a necessary and sufficient condition on~$N$ to exclude Price Manipulation Strategies. Here, we apply this result to identify which parameters in the Hawkes model exclude PMS.
We recall the notation
\begin{equation}
\alpha=\ios-\ioc = \int_{\mathbb{R}^+} (\phis-\phic)(v/m_1) \mu(\dd v),
\nonumber
\end{equation}
and define the (normalized) support of~$\mu$
\begin{equation}
\mathcal{S}(\mu) = \{ y \geq 0 \ s.t. \ \forall \varepsilon>0, \ \mu((m_1 \times y -\varepsilon, m_1 \times y +\varepsilon ))>0\}.
\nonumber
\end{equation}
%
%
\begin{proposition}\label{prop_MIHM}
The MIH model does not admit PMS if, and only if the following conditions hold
\begin{equation}\label{cond_MIHM}
\beta = \rho, \ \alpha = (1-\nu)\rho, \ \phis(x)-\phic(x)=\alpha x \text{ for } x \in \mathcal{S}(\mu) \ (\text{i.e. } m_1  I=\alpha N),
\text{ and } qD_0 = \frac{m_1}\rho \delta_0 ,
\end{equation}
or $\mu = \text{Dirac}(0)$ with $D_0=0$. 
In both cases, the optimal execution strategy is given by~\eqref{opt_strat_Pmg}. 
\end{proposition}

Note that in the case $\mu = \text{Dirac}(m_1)$ where all the jumps have the same size, one has $\mathcal{S}(\mu) = \{1\}$ thus $\phis-\phic$ is necessarily linear on $\mathcal{S}(\mu)$ and $\Delta I_t=\alpha \ \sgn(\Delta N_t)$.
If moreover $m_1=0$, we have $N\equiv 0$ and the MIH model does not depend any longer on the parameters $\alpha$ and $\beta$, that can then be fixed arbitrarily. 
\begin{proof} From Theorem~\ref{Thm_Pmg}, PMS are excluded if, and only if the price~$P$ is a martingale when $X \equiv 0$. In this case, we have from~\eqref{defP}, \eqref{dynS}, \eqref{dynD} and~\eqref{dyn_deltasigma}
$$\textup{d}P_t= - \rho D_t dt + \frac{1}{q}\textup{d} N_t= \frac{1}{q}(\textup{d} N_t-\delta_t m_1 \textup{d}t)+\left( \frac{ m_1}{q} \delta_t - \rho D_t \right) \textup{d}t.$$
Therefore, $P$ is a martingale if, and only if $\frac{ m_1}{\rho} \delta_t = q D_t $ $\Px$-a.s., $dt$ a.e. This condition is equivalent to 
$ qD_0 = \frac{m_1}\rho \delta_0 $ and $q dD_t= \frac{m_1}\rho d \delta_t$. From~\eqref{dynD} and~\eqref{dyn_deltasigma}, the latter condition is equivalent to
%
%
%
$$\rho q D_t =   \frac{m_1}\rho \beta \delta_t \text{ and }  (1-\nu)\textup{d}N_t = \frac{m_1}\rho\textup{d} I_t.$$
Using~\eqref{def_II}, the second condition is equivalent to $(1-\nu)\rho \ v = m_1 (\phis-\phic)(v/m_1)$ for all $v$ in the support of $\mu$, which implies the linearity of $\phis-\phic$ on $\mathcal{S}(\mu)$ and leads to~\eqref{cond_MIHM}. 
Conversely,~\eqref{cond_MIHM} implies $\frac{ m_1}{\rho} \delta_t = q D_t $, and $P$ is then a martingale. 
\end{proof}
%
%
\begin{remark}\label{Rk_steady_state}
When $\beta = \rho$, $\alpha = (1-\nu)\rho$, and $\phis-\phic$ is linear on $\mathcal{S}(\mu)$, we get from the previous calculations that $d( \frac{ m_1}{q} \delta_t - \rho D_t)=-\rho( \frac{ m_1}{q} \delta_t - \rho D_t)dt$, and therefore $\frac{ m_1}{q} \delta_t - \rho D_t$ converges exponentially to zero. The condition $qD_0 = \frac{m_1}\rho \delta_0$ simply means that the model starts from this steady state.  
\end{remark}
One can also check directly that the optimal strategy and its cost given by Theorem~\ref{theo:opt_strat_Hawkes}  coincide with those of Theorem~\ref{Thm_Pmg} when~\eqref{cond_MIHM} holds. For clear reasons, we call Mixed-Impact Hawkes Martingale (MIHM) model the MIH model if these conditions are satisfied. Proposition~\ref{prop_MIHM} is very interesting since it makes connections between the model parameters of the MIH model in a perfect market without PMS.
First, the condition
$\beta = \rho$ means that the mean-reverting action of liquidity providers compensates the autocorrelation in the signs of the trades of liquidity takers; we thus reach a conclusion similar to Bouchaud et al.~\cite{BGPW}. The condition $\alpha=(1-\nu) \beta$ gives a link between  the Hawkes kernel and the proportion $1-\nu$ of transient price impact. When $\iota_\textup{c}=0$, $\alpha/\beta$ represents the average number of child orders coming from one market order, and is thus equal to the proportion of endogenous orders (i.e. triggered by other orders) in the market. What we obtain here is that this ratio should be equal to $1-\nu$, which is a \textit{a priori} different measure of endogeneity, since it gives the proportion of market impact that does not influence the low-frequency price (see Section~\ref{section:low-freq_MIH}). The positivity of $\alpha$ reflects the fact that the parameter $\iota_{\textup{c}}$ tuning opportunistic trading should be small to avoid market instability.
It is interesting to notice that if~\eqref{cond_MIHM} holds, the stationarity condition $\ios+\ioc<\beta$ derived in Section~\ref{section:low-freq_MIH} is equivalent to $2 \ioc< \nu \rho$, which can be seen as a reasonable upper bound for $\ioc$. 
Last, we see that $\phis-\phic$ should be linear. Let us recall that $\phis$ and $\phic$ encode the dependence of the self-excitation (resp. the cross-excitation) effect on the volumes of incoming market orders. Condition~\eqref{cond_MIHM} implies that they should have roughly the same functional form, except for a linear part which should be stronger for $\phis$. 
However, we remind here that these conclusions are obtained in the MIH model  and should be confronted to market data. We leave this empirical investigation for further research.


Of course, in practice, it would be miraculous if the calibration of the MIH model on real financial data led to parameters satisfying exactly~\eqref{cond_MIHM}. One may rather expect these  parameters to be close but not exactly equal to those of the MIHM model, for the following reasons. First, there is no guarantee that  fitting a model to a market with no PMS leads to a model with no PMS. Second, the MIH model ignores market frictions such as the bid-ask spread and gives some advantages to the strategic trader such as the possibility to post orders immediately after the other ones (see Stoikov and Waeber~\cite{StoikovWaeber} for a study on the latency to execute an order). These facts make the existence of PMS more likely in the model than in reality. Third, we know that in practice, temporary arbitrage may exist at high frequencies. Therefore, there is no reason that fitted parameters follow exactly the MIHM condition~\eqref{cond_MIHM}. This justifies the potential practical usefulness of the strategy given by Theorem~\ref{theo:opt_strat_Hawkes} to reduce execution costs when the estimated parameters deviate from the MIHM model. Let us note that Figure~\ref{fig:Hawkes_opt_strat} illustrates such a case: all the parameters satisfy~\eqref{cond_MIHM} but~$\rho$ (which should be equal to $\beta=20$). 
The estimation of the MIH model on market data is left for future research.

The framework of the MIH model also gives some interesting insights for the characterization of the existence of short-time arbitrages. Let us introduce the following definition. 
%
%
\begin{definition}
We say that a market admits weak Price Manipulation Strategies (wPMS) if the cost of a liquidation strategy can be reduced by posting a block trade as an immediate response to a market order issued by another trader.
\end{definition}
%
%
\begin{corollary}\label{corollary:no_jumps_no_arb}
In the MIH model, the market does not admit wPMS if, and only if,
\begin{equation}\label{cond_no_wPMS}
\beta = \rho, \ \alpha = (1-\nu)\rho \text{ and } \phis(x)-\phic(x)=\alpha x \text{ for } x \in \mathcal{S}(\mu)
\end{equation}
or $\mu = \text{Dirac}(0)$.
\end{corollary}
\begin{proof}
The proof is quite straightforward from Theorem~\ref{theo:opt_strat_Hawkes}. The case $\mu = \text{Dirac}(0)$ is trivial and we consider $m_1>0$. The jump term of the strategy~\eqref{reaction_trade} should be equal to zero for any $\tau \in [0,T]$. By taking $\tau=T$, we get that $ \phis(x)-\phic(x)=(1-\nu)\rho x$  for $x \in \mathcal{S}(\mu)$. Integrating this indentity with respect to $\mu$ leads to $\alpha=\ios-\ioc=(1-\nu)\rho$. Then, from~\eqref{reaction_trade}, we should have  $(\nu \rho - \eta) \times \rho u^2 \times \omega(\eta u)=0$ for $u\in[0,T]$ which implies $\nu \rho=\eta$. Since $\eta=\beta-\alpha=\beta-(1-\nu)\rho$, we get $\beta=\rho$. The converse implication is obvious.
\end{proof}
By Remark~\ref{Rk_steady_state}, the condition $qD_0=\frac{m_1}\rho \delta_0$ means that the model has reached its equilibrium, which is basically the case after some time. Therefore, the conditions that exclude wPMS and PMS in the MIH model are quite the same. This is an interesting link between two different point of views. The condition ``no PMS'' means that there is no free source of income. The condition ``no wPMS'' rather brings on market stability, since it excludes trading volume coming from the response to other trades. Corollary~\ref{corollary:no_jumps_no_arb} is a mathematical formulation of this link in our specific model.

\subsection{The Poisson model}\label{section:Poisson_arbitrages}
Poisson processes are often used to model the arrival of the customers in queuing theory. It is therefore natural to use them to model the flow of market orders, as it has been made for example by Bayraktar and Ludkovski~\cite{BayraktarLudkovski} or Cont and de Larrard~\cite{ContLarrard} in different frameworks.

Here, in the Poisson model, $N^+$ and $N^-$ are two i.i.d. independent compound Poisson processes of respective constant jump rates $\kappa_0^+$ and $\kappa_0^-$, with the same jump law $\mu$. It is a particular case of the MIH model when $\beta=0, \ \phis\equiv 0$ and $\phic \equiv 0$, which implies $\eta = 0$. Thus, the optimal strategy and value function in this case can be deduced from Theorem~\ref{theo:opt_strat_Hawkes} (see also Remark~\ref{rk:Poisson_strat}). 

First, let us note that the Poisson model cannot satisfy the condition~\eqref{cond_MIHM}, except in the case $\rho = 0$,  where there is only permanent price impact, which is not relevant in this context. Thus, we know \textit{a priori} that PMS are possible. However, we specify in what follows that a Poisson order flow creates very simple and robust arbitrages. First, we put aside the case $\kappa_0^+ \not = \kappa_0^-$ where the trend on the price leads to obvious PMS, and consider now the more interesting case $\kappa_0^+ = \kappa_0^-$, and we simply denote by $\kappa_0$ the common intensity.

A natural choice to get a PMS is of course to consider the optimal strategy given by Theorem~\ref{theo:opt_strat_Hawkes} when liquidating $x_0=0$ assets.
A remarkable feature of this optimal strategy in the Poisson case is that it only depends on the process~$N$, and does not depend directly on the law of the jumps and their intensity. Then, when applying the optimal strategy, mainly two quantities have to be known: $qD_0$ and $\rho$. We denote by $\mathcal{C}_0(D_0)$ the cost of the optimal strategy and obtain from Theorem~\ref{theo:opt_strat_Hawkes} in this case:
\begin{equation}
(1-\epsilon) q \times \mathcal{C}_0(D_0) = - \frac {\rho T /2}{2 + \rho T} \ q^2 D_0^2
\ - \ (1-\nu)^2 \ 2\kappa_0 m_2 \left[\frac T 2 - \frac 1 \rho \ln \left( 1 + \frac {\rho T} 2 \right) \right].
\label{eqn:value_Poisson_balanced}
\end{equation}
%
In fact, PMS are very robust in this framework. The following proposition shows that even if $qD_0$ and $\rho$ are unknown, one can construct a such a strategy.
 This indicates that in our framework with a linear price impact and an exponential resilience, compound Poisson processes are not suitable to model the order flow.

\begin{proposition}\label{prop:poisson:arbitrage}
Let $\kappa_0^+=\kappa_0^-=\kappa_0>0$ and $\lambda \in (0,1)$. The following round-trip strategy $X^\lambda_0=X^\lambda_{T+}=0$ defined by 
$$ X^\lambda_{\tau+}-X^\lambda_\tau =- \frac{1-\nu}{1-\epsilon} \times  \lambda (N_\tau-N_{\tau-}) $$
at each jump of~$N$ is a PMS. Its average cost is given by
$$\E[C(X^\lambda)] =2 \lambda(1-\lambda)\frac{\kappa_0 m_2(1-\nu)^2 } {q (1-\epsilon)} \left[\frac{1-\exp(-\rho T) }{\rho}-T \right] < 0, $$
and the best choice is to take $\lambda=1/2$. 
\end{proposition}
\begin{proof}
From Remark~\ref{Rk_N_mg}, it is sufficient to focus on the case $\nu=\epsilon=0$ and $q=1$. In this case, we have 
$$ C(X)= \int_{[0,T)} D_u \ \textup{d}X^\lambda_u \ + \ \frac{1}{2} \underset{0\leq\tau<T}{\sum} (\Delta X^\lambda_\tau)^2 \ - \ D_T X^\lambda_T \ + \ \frac{1}{2} \ (X^\lambda_T)^2  , $$
with $D_t=D_0 + \int_0^t \exp(-\rho(t-s))dN_s + \int_0^t \exp(-\rho(t-s))dX^\lambda_s. $
From $ \int_{[0,T)} D_u  \textup{d}X^\lambda_u =  - \lambda \underset{0\leq\tau<T}{\sum} \left[  D_{\tau-}\Delta N_\tau + (\Delta N_\tau)^2 \right]$, we get $\E[\int_{[0,T)} D_u \ \textup{d}X^\lambda_u]= - \lambda \E[ \underset{0\leq\tau<T}{\sum} (\Delta N_\tau)^2]  =- 2 \lambda \kappa_0 m_2 T $. Since $X^\lambda_T=-\lambda N_T$ and $D_T=D_0 +(1-\lambda) \int_0^T \exp(-\rho(T-s))dN_s$ a.s., we have $\E[(X^\lambda_T)^2]=2 \lambda^2 \kappa_0 m_2 T$ and 
$$\E[-D_T X^\lambda_T]= \lambda(1-\lambda) \E \left[N_T \int_0^T \exp(-\rho(T-s))dN_s \right]= 2  \lambda(1-\lambda) \kappa_0 m_2  \frac{1-\exp(-\rho T)}{\rho}.$$
This eventually yields
\begin{align*}
\E[C(X^\lambda)]& =- 2 \lambda \kappa_0 m_2 T + \lambda^2 \kappa_0 m_2 T  + 2  \lambda(1-\lambda)  \kappa_0 m_2  \frac{1-\exp(-\rho T)}{\rho}+\lambda^2  \kappa_0 m_2 T \\& = 2  \lambda(1-\lambda)  \kappa_0 m_2  \left(\frac{1-\exp(-\rho T)}{\rho} -T \right). 
\end{align*}
\end{proof}

\appendix

\section{Explicit formulas for the optimal strategy}\label{appendix:formulas_opt_strat}

We use the function
\begin{equation}\label{def_fctL}
L(r,\lambda,t) \ := \ r \int_0^t \frac{\exp(\lambda s)}{2+r s} \ \textup{d}s
\ = \ \exp(-2\lambda/r) \ \left[\mathcal{E}\left( \frac\lambda r(2+r t) \right)
-\mathcal{E}\left( \frac{2\lambda}r \right) \right],
\end{equation}
where $\mathcal{E}(y) = - \int_{-y}^{+\infty} \frac{e^{-u}} u \ \textup{d}u$ is the exponential integral of $y$, in terms of Cauchy principal value if $y>0$. Since we only consider differences $\mathcal{E}(y)-\mathcal{E}(y')$ with either $y,y'>0$ or $y,y'<0$, we will only consider proper integrals. The function~$\mathcal{E}$ is standard and is implemented in many packages such as the Boost C++ library. Thus, $L$ can be evaluated as a closed formula. 

We refer to~\eqref{def_zeta} and~\eqref{eqn:def_omega} for the definitions of $\zeta$ and $\omega$.

%
\textbf{Auxiliary functions:}
%
For $0 \leq s \leq t \leq T$,
\begin{align*}
\phi_\eta(t) = \ \frac1 {2(2 + \rho (T-t))} \times 
& \Big[1 +\exp(-\eta(T-t))+\nu \rho (T-t) \zeta(\eta(T-t)) \\
& + \frac\beta\rho
\left[2+\rho (T-t) \times \left\{ 1+ \zeta(\eta (T-t))+ \nu \rho (T-t) \ \omega(\eta(T-t)) \right\}\right]\Big],
\end{align*}
\begin{eqnarray}
\Phi_0(s,t)  & = & \left[\frac \beta \rho + \frac\nu2 \left( \frac12 -\frac\beta\rho \right) \right] \times
\frac {\exp(-\beta s) - \exp(-\beta t)} \beta
\nonumber \\
& & \ + \ (1-\nu) \left( 1 - \frac \beta \rho \right) \times
\frac{\exp(-\beta T)}\rho \times
\left[ L(\rho,\beta,T-s) - L(\rho,\beta,T-t) \right]
\nonumber \\
& & \ + \ \frac\nu4 \left[
(T-s) \exp(-\beta s) - (T-t) \exp(-\beta t)
\right],
\nonumber
\end{eqnarray}
and for $\eta \neq 0$,
\begin{eqnarray}
\Phi_\eta(s,t) & = &
\frac12 \left( \frac1\rho+\frac\nu\eta\right) \times [\exp(-\beta s)-\exp(-\beta t)]
\nonumber \\
& & \quad + \ \frac{\exp(-\beta T)}{2\rho} \times
\left[ 1+\frac{\nu(\rho-2\beta)}\eta+\frac\beta\eta\left(1-\frac{\nu\rho}\eta\right) \right] \times
[L(\rho,\beta,T-s)-L(\rho,\beta,T-t) ]
\nonumber \\
& & \quad + \ \frac{\exp(-\beta T)}{2\rho} \times
\left[ 1-\frac{\nu\rho}\eta-\frac\beta\eta\left(1-\frac{\nu\rho}\eta\right) \right] \times
[L(\rho,\alpha,T-s)-L(\rho,\alpha,T-t) ].
\nonumber
\end{eqnarray}

We now give the explicit formulas for the whole optimal strategy. They are valid for all $\eta \in \mathbb{R}$.

\textbf{Trend strategy:}
\begin{eqnarray}
(1-\epsilon) \Delta X^\textup{trend}_0 & = & 
\frac{\frac{\delta_0 m_1}{2\rho} \times \left[2+\rho T \times \left\{ 1+ \zeta(\eta T)+ \nu \rho T \ \omega(\eta T) \right\}\right]
- [1+\rho T] q D_0}{2+\rho T},
\nonumber \\
(1-\epsilon) \Delta X^\textup{trend}_T & = & 
\frac {\delta_0 m_1}{2\rho} \times \left[ \frac {2+\rho T \times \left\{ 1+ \zeta(\eta T)+ \nu \rho T \ \omega(\eta T) \right\}} {2 + \rho T}
- 2 \rho \ \Phi_\eta(0,T) - 2 \exp(-\beta T)  \right]
\nonumber \\
& & \quad + \ \frac {q D_0} {2+\rho T}, 
\label{Xtrend_eta}
\end{eqnarray}
and, on $(0,T)$,
\begin{eqnarray}
(1-\epsilon) \textup{d} X^\textup{trend}_t & = & 
\frac{\delta_0 m_1}{2\rho} \times \left[ \frac {2+\rho T \times \left\{ 1+ \zeta(\eta T)+ \nu \rho T \ \omega(\eta T) \right\}} {2 + \rho T}
- 2 \rho \ \Phi_\eta(0,t) - 2 \phi_\eta(t) \exp(-\beta t)
\right] \ \rho \textup{d}t
\nonumber \\
& & \ + \
\frac {q D_0} {2 + \rho T} \ \rho \textup{d}t.
\nonumber
\end{eqnarray}

\textbf{Dynamic strategy:}
\begin{eqnarray}
(1-\epsilon) \Delta X^\textup{dyn}_0 & = & 0,
\nonumber \\
(1-\epsilon) \Delta X^\textup{dyn}_T & = &  - \ m_1 \left[ \Theta_{\chi_T} \ \Phi_\eta\left({\tau_{\chi_T}},T\right)
\ + \ \overset{\chi_T-1}{\underset{i=1}{\sum}}
\Theta_i \ \Phi_\eta(\tau_i,\tau_{i+1})  \right]
\ + \ \underset{0<\tau\leq T}{\sum} \frac{(1-\nu) \ \Delta N_\tau}{2 + \rho (T-\tau)} 
\nonumber \\
& & \quad 
+ \ \frac{m_1}{2\rho} \times \underset{0<\tau\leq T}{\sum} 
\frac{2+\rho (T-\tau) \times \left\{ 1+ \zeta(\eta (T-\tau))+ \nu \rho (T-\tau) \ \omega(\eta(T-\tau)) \right\}}{2+\rho(T-\tau)}\ \Delta I_\tau
\nonumber\\
& & \quad - \ \frac{m_1}\rho \ \Theta_{\chi_T} \exp(-\beta T), 
\label{Xdyn_eta}
\end{eqnarray}
and, on $(0,T)$,
\begin{eqnarray}
(1-\epsilon) \textup{d} X^\textup{dyn}_t & = & \ - \ m_1 \ \phi_\eta(t) \ \Theta_{\chi_t} \exp(-\beta t) \ \textup{d}t
\ + \ \left[\underset{0<\tau\leq t}{\sum} \ \frac{(1-\nu) \Delta N_\tau}{2 + \rho (T-\tau)} \right] \rho \textup{d}t \nonumber \\
& & \ + 
\left[ \underset{0<\tau\leq t}{\sum} 
\frac{2+\rho (T-\tau) \times \left\{ 1+ \zeta(\eta (T-\tau))+ \nu \rho (T-\tau) \ \omega(\eta(T-\tau)) \right\}}{2+\rho(T-\tau)}\ \Delta I_\tau \right]
 \ \frac{m_1}2 \ \textup{d}t
\nonumber \\
& & \ - \ \left[ \Theta_{\chi_t} \ \Phi_\eta\left({\tau_{\chi_t}},t\right)
\ + \ \overset{\chi_t-1}{\underset{i=1}{\sum}}
\Theta_i \ \Phi_\eta(\tau_i,\tau_{i+1})  \right] \rho m_1 \ \textup{d}t
\nonumber \\
& & \nonumber \\
& & \hspace{-1.5cm} + \  \frac{1 + \rho (T-t)}{2 + \rho (T-t)}
 \left\{\frac {m_1} \rho \ \textup{d} I_t - (1-\nu) \ \textup{d} N_t \right\}
\ + \ \frac{m_1}{2\rho} (\nu \rho-\eta) \times \frac{\rho (T-t)^2 \times \ \omega(\eta(T-t)) }{2+\rho(T-t)} \ \textup{d}I_t.
\nonumber
\end{eqnarray}
\section{Proof for the optimal control problem (results of Theorem~\ref{theo:opt_strat_Hawkes} and Appendix~\ref{appendix:formulas_opt_strat})}\label{appendix:proof_opt_ctrl}

\subsection{Notations and methodology}

The jump intensity of the process $(N_t)$ is characterized by the c\`adl\`ag Markovian process $(\delta_t, \Sigma_t)$ defined by~\eqref{dyn_deltasigma}, taking values in $\mathbb{R} \times \mathbb{R}^+$. The state variable of the problem is then $(X_t, D_t, S_t, \delta_t, \Sigma_t)$, and the control is $X_t - x_0$, i.e. the variation of the position of the strategic trader, $(X_t)_{t \in [0,T]}$ being an admissible strategy as described in Definition \ref{def:liq_strat}. The control program is thus to minimize $\mathbb{E} \left[ C(0,X) \right]$ over all admissible strategies, where the cost $C(t,X)$ of the strategy $X$ between $t$ and $T$ is given by
\begin{equation}
 C(t,X) \ = \ \int_{[t,T)} P_u \ \textup{d}X_u 
\ + \ \frac1{2q} \underset{t\leq\tau<T}{\sum} (\Delta X_\tau)^2
 \ - \ P_T X_T \ + \ \frac1{2q} \ X_T^2 .
\nonumber
\end{equation}
The final value at time $t=T$ is the cost of a market order of signed volume $\Delta X_T = - X_T$ 
(so that $X_{T^+} = X_T + \Delta X_T = 0$). At time $t$, the price $P_t$ depends on $D_t$ and $S_t$ which in turn depend on $(X_u)_{u \in [0,t]}$. Let us define $\mathcal{A}_t$ the set of admissible strategies on $[t,T]$, with $t \in [0,T]$. The value function of the problem is
\begin{equation}
\mathcal{C}(t,x,d,z,\delta,\Sigma)
\ = \
\underset{X \in \mathcal{A}_t}{\inf} \ 
\mathbb{E} \left[ C(t,X) \right]
\nonumber
\end{equation}
with $X_t = x, \ D_t = d, \ S_t = z$, $\delta_t = \delta$ and $\Sigma_t = \Sigma$. In order to determine analytically the value function and the optimal control of the problem, we use the probabilistic formulation of the verification theorem. We determine \textit{a priori} a continuously differentiable function $\mathcal{C}(t,x,d,z,\delta,\Sigma)$ and an admissible strategy $X^*$ and then we verify that
\begin{equation}
\Pi_t(X) \ := \  \int_0^t P_u \ \textup{d}X_u
\ + \ \frac1{2q} \underset{0\leq\tau< t}{\sum} (\Delta X_\tau)^2
\ + \ \mathcal{C}(t,X_t,D_t,S_t,\delta_t,\Sigma_t)
\label{def:Pi_t}
\end{equation}
is a submartingale for any admissible strategy $X$, and that $\Pi_t(X^*)$ is a martingale. We proceed in three steps:
\begin{enumerate}
\item We define a suitable function $\mathcal{C}$, and derive a set of ODEs on its coefficients which is a necessary condition for $\mathcal{C}$ to be the value function of the problem. 
\item We solve the set of ODEs. 
\item Using the results of the previous steps, we derive the strategy $X^*$ such that $\Pi_t(X^*)$ is a martingale. 
\end{enumerate}
The verification argument then yields that $\mathcal{C}(t,x,d,z,\delta,\Sigma)$ is the value function and that $X^*$ is optimal. Without loss of generality, we can assume that $q=1$ by using Remark~\ref{Rk_Cost}. 

\subsection{Necessary conditions on the value function}\label{appendix:nec_cond_proof_opt}

We search a cost function $\mathcal{C}$ as a generic quadratic form of the variables $x,d,z,\delta,\Sigma$ with time-dependent coefficient (the variable $z$ symbolizes the current value of the fundamental price $S_t$). As we see further, we need $\mathcal{C}$ to verify
$\partial_x \mathcal{C}+ (1-\epsilon) \partial_d \mathcal{C} + \epsilon \ \partial_z \mathcal{C} + d + z = 0$ : it is thus necessary that $\mathcal{C}$ is a quadratic form of $(d - (1-\epsilon)x), \ (z - \epsilon x), \ \delta$ and $\Sigma$, plus a term $-(d+z)^2/2$. We define
\begin{eqnarray}
\mathcal{C}(t,x,d,z,\delta,\Sigma) & = & a(T-t) (d-(1-\epsilon)x)^2 
\ + \ \frac12 (z - \epsilon x)^2 \ + \ (d-(1-\epsilon)x)(z - \epsilon x) \ - \ \frac{(d+z)^2}2 
\nonumber \\
& & \quad + \ b(T-t) \ \delta \ (d-(1-\epsilon)x)
\ + \ c(T-t) \ \delta^2 \ + \ e(T-t) \ \Sigma \ + \ g(T-t),
\label{eqn:cost_form_Hawkes}
\end{eqnarray}
with $a,b,c,e,g: \mathbb{R^+} \rightarrow \mathbb{R}$ continuously differentiable functions. We choose the limit condition $\mathcal{C}(T,x,d,z,\delta,\Sigma) \ = \ -(d+z)x + x^2/2 = \frac12(d+z-x)^2 - (d+z)^2/2$, which is the cost of a trade of signed volume $-x$. We thus have 
$$a(0) = \frac12, \ b(0) = c(0) = e(0) = g(0) = 0.$$ Let us note that other terms should be added in equation (\ref{eqn:cost_form_Hawkes}) for $\mathcal{C}$ to be a generic quadratic form. The five terms
\begin{eqnarray}
h_1(T-t) \ (d-(1-\epsilon)x) \ + \ h_2(T-t) \ \Sigma (d-(1-\epsilon)x) \ + \ h_3(T-t) \ \delta \Sigma \ + \ h_4(T-t) \ \delta \  + \ h_5(T-t) (z-\epsilon x)
\nonumber
\end{eqnarray}
have to be equal to zero since $\mathcal{C}(t,x,d,z,\delta,\Sigma)=\mathcal{C}(t,-x,-d,-z,-\delta,\Sigma)$ by using Remark~\ref{Rk_Cost} and the fact that the buy and sell orders play a symmetric role. 
%
For the term in $\Sigma^2$, we checked in prior calculations that it is necessarily associated to a zero coefficient. 
For $\Delta x \in \mathbb{R}$, we have
\begin{equation}
\mathcal{C}(t,x+\Delta x,d+(1-\epsilon) \Delta x, z+\epsilon \Delta x,\delta,\Sigma) - \mathcal{C}(t,x,d,z,\delta,\Sigma)
 \ = \  -(d+z) \times \Delta x \ - \ \frac{(\Delta x)^2}2.
\label{cost_if_deltaX}
\end{equation}
In what follows, we drop the dependence of $\mathcal{C}(t,X_t,D_t,S_t,\delta_t,\Sigma_t)$ on
$(t,X_t,D_t,S_t,\delta_t,\Sigma_t)$ to obtain less cumbersome expressions.
The process $ \mathcal{C}(t,X_t,D_t,S_t,\delta_t,\Sigma_t)$ is l\`adl\`ag, and with the notations of Remark~\ref{notations_ladlag}, we have by using~\eqref{cost_if_deltaX}
%
%
\begin{eqnarray}
\textup{d} \mathcal{C} & = & \partial_t \mathcal{C} \ \textup{d}t
\ + \ \partial_x \mathcal{C} \ \textup{d}X^\text{c}_t 
\ + \ \partial_d \mathcal{C} \Big(- \rho D_t \textup{d}t + (1-\epsilon) \textup{d}X^\text{c}_t \Big) 
\ + \ \partial_z \mathcal{C}\ \epsilon \textup{d}X^\text{c}_t
\nonumber \\
& & \ - \ \beta \ \delta_t \ \partial_\delta \mathcal{C} \ \textup{d}t
 -  \beta  (\Sigma_t-2\kappa_\infty) \ \partial_\Sigma \mathcal{C} \ \textup{d}t
\nonumber \\
& & \hspace{-0.5cm} +\Big[\mathcal{C}(t,X_t,D_{t^-}+ (1-\nu) \Delta N_t, S_{t^-} + \nu \Delta N_t,
\delta_{t^-}+\Delta I_t,\Sigma_{t^-}+\Delta \overline{I}_t) 
- \  \mathcal{C}(t,X_t,D_{t^-},S_{t^-},\delta_{t^-},\Sigma_{t^-}) \Big]
\nonumber \\
& & \ - \ (D_t+S_t) \ \Delta X_t \ - \ \frac{(\Delta X_t)^2}2.
\nonumber
\end{eqnarray}
where we refer to (\ref{def_II}) for the definitions of $I$ and $\overline{I}$.
The definition of $\Pi(X)$ given by (\ref{def:Pi_t}) yields
$\textup{d}\Pi_t(X) = (D_t+S_t) \textup{d}X^\text{c}_t + (D_t+S_t) \Delta X_t + (\Delta X_t)^2/2 + \textup{d}\mathcal{C}$.
We define the continuous finite variation process $(A^X_t)_{t \in (0,T)}$ such that $A^X_{0^+} = \mathcal{C}(0,X_{0^+},D_{0^+},S_{0^+},\delta_0,\Sigma_0)$ and
for $t \in (0,T)$
\begin{eqnarray}
\textup{d}A^X_t & = & (D_t+S_t) \ \textup{d}X^\text{c}_t \ + \ Z(t,X_t,D_t,S_t,\delta_t,\Sigma_t) \textup{d}t
\nonumber \\
& & \qquad + \  \partial_t \mathcal{C} \ \textup{d}t
\ + \ \partial_x \mathcal{C}\ \textup{d}X^\text{c}_t 
\ + \ \partial_d \mathcal{C} \ \Big(- \rho D_t \textup{d}t + (1-\epsilon) \textup{d}X^\text{c}_t \Big) 
\ + \ \partial_z \mathcal{C} \ \epsilon \textup{d}X^\text{c}_t
\nonumber \\
& & \qquad - \ \beta \ \delta_t \ \partial_\delta \mathcal{C} \ \textup{d}t
\ - \ \beta \ (\Sigma_t-2\kappa_\infty) \ \partial_\Sigma \mathcal{C} \ \textup{d}t,
\nonumber
\end{eqnarray}
where, for $V \sim \mu$, 
%
%

$
Z(t,x,d,z,\delta,\Sigma) := 
$
\begin{align}
\frac{\Sigma + \delta}2 &\times 
\mathbb{E} \big[ \mathcal{C}(t,x,d + (1-\nu) V, z+\nu V,\delta+(\phis-\phic)(V/m_1),\Sigma+(\phis+\phic)(V/m_1)) 
-  \mathcal{C}(t,x,d,z,\delta,\Sigma) \big]
\nonumber \\
+ \ \frac{\Sigma - \delta}2 &\times 
\mathbb{E} \big[ \mathcal{C}(t,x,d - (1-\nu)V, z-\nu V,\delta-(\phis-\phic)(V/m_1),\Sigma+(\phis+\phic)(V/m_1)) 
-  \mathcal{C}(t,x,d,z,\delta,\Sigma) \big].
\nonumber
\end{align}
Then, $\Pi(X) - A^X$ is a martingale (let us note that almost surely, $\textup{d}t$ -a.e. on $(0,T)$, $Z(t,X_t,D_{t^-},S_{t^-},\delta_{t^-},\Sigma_{t^-}) = Z(t,X_t,D_t,S_t,\delta_t,\Sigma_t)$). This yields that $\Pi(X)$ is a submartingale (resp. a martingale) iff $A^X$ is increasing (resp. constant). From~\eqref{cost_if_deltaX}, we obtain
$\partial_x \mathcal{C}(t,x,d,z,\delta,\Sigma) + (1-\epsilon) \partial_d \mathcal{C}(t,x,d,z,\delta,\Sigma) + \epsilon \ \partial_z \mathcal{C}(t,x,d,z,\delta,\Sigma)  + d + z = 0$, and then 
\begin{equation}
\textup{d}A^X_t \ = \ \Big\{
\partial_t \mathcal{C} - \rho \ D_t \ \partial_d \mathcal{C}
+ Z(t,X_t,D_t,S_t,\delta_t,\Sigma_t)
\ - \beta \ \delta_t \ \partial_\delta \mathcal{C} \  
 - \beta \ (\Sigma_t-2\kappa_\infty) \ \partial_\Sigma \mathcal{C}
\Big\} \textup{d}t.
\label{eqn:dA1_Hawkes}
\end{equation}
Given the quadratic nature of the problem, we search a process $A^X$ of the form
\begin{equation}
\textup{d}A^X_t \ = \ \frac\rho{1-\epsilon} \textup{d}t \times
 \Big[
j(T-t) (D_t-(1-\epsilon)X_t) \ - \ D_t \ + \ k(T-t) \ \delta_t
\Big]^2,
\label{eqn:dA2_Hawkes}
\end{equation}
with $j, k : \mathbb{R^+} \rightarrow \mathbb{R}$ continuously differentiable functions, in order to obtain an non-decreasing process $A^X$ that can be constant for a specific strategy $X^*$. Let us note $Y_t := D_t - (1-\epsilon)X_t$, $\ \Xi_t := S_t - \epsilon X_t$, $ \ y := d-(1-\epsilon)x$, $ \ \xi := z-\epsilon x$. 
Since $d+z = y + \xi + x = \xi + \frac{d-\epsilon y}{1-\epsilon}$, we have
\begin{eqnarray}
\partial_t \mathcal{C}(t,x,d,z,\delta,\Sigma) & = & - \dot a \ y^2
\ - \ \dot b \ \delta y \ - \ \dot c \ \delta^2  \ - \ \dot e \ \Sigma \ - \ \dot g
\nonumber , \\
- \rho d \ \partial_d \mathcal{C}(t,x,d,z,\delta,\Sigma) & = & - \left(2 \rho a + \frac{\rho \epsilon}{1-\epsilon}\right) \ d y
\ + \ \frac\rho{1-\epsilon} \ d^2 \ - \ \rho b \ \delta d
\nonumber ,\\
- \beta \delta \ \partial_\delta \mathcal{C}(t,x,d,z,\delta,\Sigma) 
& = & - \beta b \ \delta y \ - \ 2 \beta c \ \delta^2 
\nonumber ,\\
- \beta (\Sigma - 2 \kappa_\infty) \ \partial_\Sigma \mathcal{C}(t,x,d,z,\delta,\Sigma) & = &
- \beta e \ \Sigma \ + \ 2 \beta \kappa_\infty e,
\nonumber
\end{eqnarray}

%
Let $V\sim\mu$. One has 
\begin{equation}
\E[(\phis-\phic)(V/m_1)]=\ios-\ioc=\alpha 
\quad , \quad
 \E[(\phis+\phic)(V/m_1)]=\ios+\ioc=\alpha+2\ioc.
\nonumber
\end{equation}
Thus,

$
\mathbb{E} \big[ \mathcal{C}(t,x,d + (1-\nu) V, z+\nu V,\delta+(\phis-\phic)(V/m_1),\Sigma+(\phis+\phic)(V/m_1)) 
-  \mathcal{C}(t,x,d,z,\delta,\Sigma) \big]
$
\begin{eqnarray}
& = & a \ [(1-\nu)^2 m_2 + 2(1-\nu)m_1 \ y] \ + \ \frac{\nu^2}2 m_2 \ + \ \nu m_1 \ \xi
\nonumber \\
& & \quad + \ \nu(1-\nu)m_2 \ + \ \nu m_1 y \ + \ (1-\nu) m_1 \xi
\ - \frac 1 2 \left( m_2 + 2 \ m_1 \ \xi + \frac{2m_1}{1-\epsilon} \ d 
- \frac{2 \epsilon m_1}{1-\epsilon} \ y \right) 
\nonumber \\
& & \quad + \ b \ [(1-\nu)m_1 \ \delta \ + \ \alpha \ y \ + \ \tilde{\alpha}(1-\nu)]
\ + \ c \ [\alpha_2 \ + \ 2 \alpha \ \delta] \ + \ (\alpha +2 \iota_\textup{c}) e,
\nonumber
\end{eqnarray}
%
with
\begin{equation}
\tilde{\alpha} = \E[V \times (\phis-\phic)(V/m_1)]
\quad , \quad
\alpha_2 = \E[(\phis-\phic)^2(V/m_1)].
\label{def_alpha_t_2}
\end{equation}
These quantities $\tilde{\alpha}$ and $\alpha_2$ are finite by assumption. This gives
\begin{eqnarray}
 Z(t,x,d,z,\delta,\Sigma) & = & \left( m_1 \times \left[2 (1-\nu) a + \nu +\frac\epsilon{1-\epsilon}\right] 
\ + \ \alpha b \right) \ \delta y 
\ - \ \frac{m_1}{1-\epsilon} \ \delta d
\nonumber \\
& & \ + \left[ (1-\nu) m_1 b \ + \ 2 \alpha c \right] \ \delta^2
\nonumber \\
& & + \ \left( m_2 \times \left[ (1-\nu)^2 a + \nu(1-\nu/2) - \frac12 \right] + \tilde{\alpha} (1-\nu) b 
+ \alpha_2 c + (\alpha +2 \iota_\textup{c}) e \right) \ \Sigma
\nonumber ,
\end{eqnarray}
where we consider $\mathcal{C}$ as a function of the variables $t,x,d,z,\delta,\Sigma$ as in equation (\ref{eqn:dA1_Hawkes}), and substitute $d - (1-\epsilon)x$ by $y$ and $z - \epsilon x$ by $\xi$ in the results. We then make the change of variables $(x,d,z,\delta,\Sigma) \rightarrow (y,d,\xi,\delta,\Sigma)$, and we identify each term of equations (\ref{eqn:dA1_Hawkes}) and (\ref{eqn:dA2_Hawkes}):

\textbf{(Eq. $dy$):} $\quad
 - \left(2 \rho a + \frac{\rho \epsilon}{1-\epsilon}\right) = - \frac{2\rho}{1-\epsilon} j$.

\textbf{(Eq. $y^2$):} $\quad
 - \dot a \ = \ \frac\rho{1-\epsilon}j^2$.

(Eq. $dy$) yields $j = (1-\epsilon)a + \frac\epsilon2$. We input this relation in (Eq. $y^2$) and we have
$
\dot j = (1-\epsilon) \dot a = - \rho j^2
$
thus
$
j(u) = \frac1{2+\rho u}
$
since $j(0) = (1-\epsilon) a(0) + \frac\epsilon2 = \frac12$. This yields $a(u) = \frac1{1-\epsilon} \left( \frac1{2+\rho u} - \frac\epsilon2 \right)$ with (Eq. $dy$).

\textbf{(Eq. $\delta y$):} $\quad
 - \ \dot b \ - \ \beta b \ + \ \alpha b 
\ + \ m_1 \times \left[2 (1-\nu) a + \nu +\frac\epsilon{1-\epsilon}\right]
\ = \ \frac{2 \rho}{1-\epsilon} j k
$.

\textbf{(Eq. $\delta d$):} $\quad
- \ \rho b \ - \ \frac{m_1}{1-\epsilon} \ = \ -\frac{2\rho}{1-\epsilon} k
$,

which yields
$
k(u) \ = \ \frac{1-\epsilon}2 \ b(u) \ + \ \frac{m_1}{2\rho}
$.
Plugging equation (\ref{eqn:k_general}) in (Eq. $\delta y$), we have
$
\dot b = -(\beta-\alpha) b 
- \frac{2\rho}{1-\epsilon} j \left( \frac{1-\epsilon}2 b + \frac{m_1}{2\rho}\right)
+ m_1 \left[2 (1-\nu) a + \nu +\frac\epsilon{1-\epsilon}\right]
$, and since $j/(1-\epsilon) \ = \ a + \epsilon/[2(1-\epsilon)]$,
we have

$
\dot b(u) \ = \
\left[ -(\beta-\alpha) - \frac\rho{2+\rho u} \right] b(u)
\ + \ \frac{m_1}{1-\epsilon} \times \frac{1+\nu \rho u}{2+\rho u}
$.

\textbf{(Eq. $\delta^2$):}$ \quad
- \ \dot c \ - \ 2 \beta c \ + \ 2 \alpha c \ + \ (1-\nu) m_1 b \ = \ \frac\rho{1-\epsilon} \ k^2
$.

\textbf{(Eq. $\Sigma$):}$ \quad
- \ \dot e \ - \ \beta e \ + \ (\alpha + 2 \iota_\textup{c} )e \ + \ 
m_2 \times \left[ (1-\nu)^2 a + \nu(1-\nu/2) - \frac12 \right] \ + \ \tilde{\alpha} (1-\nu)  b \ + \ \alpha_2 c
\ = \ 0
$.

We have $2 (1-\epsilon) \times \left[ (1-\nu)^2 a + \nu(1-\nu/2) - \frac12 \right]
= 2 (1-\nu)^2/(2+\rho u) - (1-\nu)^2 \epsilon + \nu(2-\nu)(1-\epsilon) - (1-\epsilon) $,
thus

$
\dot e(u) \ = \ -(\beta-\alpha- 2 \iota_\textup{c}) e(u) \ + \ \tilde{\alpha} (1-\nu) b(u)
 \ + \ \alpha_2 c(u)
\ + \ \frac{(1-\nu)^2 \ m_2}{1-\epsilon} \times \left[ \frac1{2+\rho u} - \frac12 \right]
$

\textbf{(Eq. constant):}$ \quad
- \ \dot g \ + \ 2 \beta \kappa_\infty e \ = \ 0
$.

We obtain two conditions on the coefficients of the process $A^X$
\begin{eqnarray}
j(u) & = & \frac1{2+\rho u},
\label{eqn:j_general} \\
k(u) & = & \frac{1-\epsilon}2 \ b(u) \ + \ \frac{m_1}{2\rho},
\label{eqn:k_general}
\end{eqnarray}
and the following set of necessary conditions on the coefficients of $\mathcal{C}$
%
%
\begin{eqnarray}
a(u) & = & \frac1{1-\epsilon} \left( \frac1{2+\rho u} - \frac\epsilon2 \right)
\label{eqn:a_general}, \\
\dot b(u) & = &
\left[ -(\beta-\alpha) - \frac\rho{2+\rho u} \right] b(u)
\ + \ \frac{m_1}{1-\epsilon} \times \frac{1+\nu \rho u}{2+\rho u}
\label{eqn:b_general}, \\
\dot c(u) & = & -2(\beta-\alpha) \ c(u) \ + \ (1-\nu) m_1 \ b(u)
 \ - \ \frac\rho{1-\epsilon} \ k(u)^2
\label{eqn:c_general}, \\
\dot e(u) & = & -(\beta-\alpha -2 \iota_\textup{c}) e(u) \ + \ \tilde{\alpha} (1-\nu) b(u)
 \ + \ \alpha_2 c(u)
\ + \ \frac{(1-\nu)^2 \ m_2}{1-\epsilon} \times \left[ \frac1{2+\rho u} - \frac12 \right],
\label{eqn:e_general} \\
\dot g(u) & = & 2\beta\kappa_\infty \ e(u),
\label{eqn:g_general} \\
b(0) & = & c(0) \ = \ e(0) \ = \ g(0) \ = \ 0.
\nonumber
\end{eqnarray}
The resolution of this set of equations determines entirely the function $\mathcal{C}(t,x,d,z,\delta,\Sigma)$ defined in (\ref{eqn:cost_form_Hawkes}). This is the purpose of the next step of this proof. Let us note that at this stage, we already know that the system given by Equations (\ref{eqn:j_general}) to (\ref{eqn:g_general}) admits a unique solution, and that the function $\mathcal{C}$ which solves the system is the value function of the problem by using the verification argument.

\subsection{Resolution of the system of ODEs}\label{appendix:resolution_proof_opt}

First of all, we use Equation (\ref{eqn:a_general}) to simplify the function $\mathcal{C}$. The constant term (w.r.t. the time variable $t$) in equation (\ref{eqn:cost_form_Hawkes}) is
$\frac12 (z - \epsilon x)^2 \ + \ (d-(1-\epsilon)x)(z - \epsilon x) \ - \ \frac{(d+z)^2}2  \ = \
 - z x - \frac{d^2}2 - \epsilon d x + \left[ \frac\epsilon2 + \frac\epsilon2 (1-\epsilon) \right] x^2
$, thus the sum of  $a(T-t) (d-(1-\epsilon)x)^2$ and this constant term can be rewritten as
\begin{equation}
- (z+d) x
\ + \ \left[ \frac{1-\epsilon}{2+\rho(T-t)} + \frac\epsilon 2 \right] x^2
\ - \ \frac1{1-\epsilon} \times \frac{\rho(T-t)/2}{2+\rho(T-t)} \ d^2
\ + \ \frac{\rho (T-t)}{2+\rho(T-t)} \ d x.
\label{eqn:simplif_a_const}
\end{equation}
We note $\eta = \beta - \alpha$. To solve equation (\ref{eqn:b_general}), we search a solution of the form $b(u) =  \tilde{b}(u) \times \exp(-\eta u)/(2+\rho u)$. This yields
$\dot{\tilde{b}}(u) =  \frac{m_1}{1-\epsilon} \times (1+\nu \rho u) \times \exp(\eta u)$.
Using the respective definitions~\eqref{def_zeta} and~\eqref{eqn:def_omega} of the functions $\zeta$ and $\omega$, it is easy to see that for all $\eta \in \mathbb{R}$,
\begin{equation}
\exp(-\eta u) \int_0^u (1+\nu \rho s) \ \exp(\eta s) \ \text{d}s \ = \ u \zeta(\eta u) + \nu \rho u^2 \omega(\eta u).
\nonumber
\end{equation}
Since $\tilde{b}(0) = 2 b(0) = 0$,
we obtain
\begin{equation}
b(u) \ = \ \frac{m_1 u}{1-\epsilon} \times \frac{\zeta(\eta u)+ \nu \rho u \ \omega(\eta u)}{2+\rho u}
\ = \ \frac1{1-\epsilon} \times \frac{\rho u}{2+\rho u} \times \frac{m_1}\rho \ \mathcal{G}_\eta(u),
\label{eqn:b_simplified_aneqb}
\end{equation}
where
\begin{equation}
\mathcal{G}_\eta(u) := \zeta(\eta u)+ \nu \rho u \ \omega(\eta u).
\nonumber
\end{equation}
Equation (\ref{eqn:k_general}) then gives
%
\begin{equation}
k(u) \ = \ \frac{m_1}{2\rho} \times \frac{2+\rho u \times \left\{ 1+ \zeta(\eta u)+ \nu \rho u \ \omega(\eta u) \right\}}{2+\rho u}.
\label{eqn:k_aneqb}
\end{equation}
The remaining functions $c$, $e$ and $g$ do not play any role to determine the optimal strategy, and their expressions are harder to obtain. Let us first consider the case $\eta \neq 0$.
After some tedious calculations, we can show that the function $c$ that solves~\eqref{eqn:c_general} with  $c(0) =0$ is given by:
\begin{equation}
c(u) = - \ \frac1{1-\epsilon} \times \frac{\rho u / 2}{2+\rho u} \times \frac{m_1^2}{\rho^2} \ \mathcal{G}_\eta(u)^2
\ - \ \frac{m_1^2}{8(1-\epsilon)\rho} \times\left( 1-\frac{\nu \rho}\eta \right)^2 \times u \zeta(\eta u)
\times \left[ 1+\exp(-\eta u) -2 \zeta(\eta u) \right].
\label{eqn:c_Hawkes_noncrit}
\end{equation}
For the functions $e$ and $g$, we recall here that they admit explicit but very cumbersome formulas that can be obtained by using a formal calculus software. 
In the case $\eta = 0$, the resolution of the ODEs is easier, and we get
\begin{align}
c(u)  & = \
- \ \frac{(1-\nu)^2}{1-\epsilon} \times \frac{m_1^2}{\rho^2} \times \left[ \frac12 - \frac1{2+\rho u} \right]
\ - \ \frac{\nu m_1^2}{\rho^2(1-\epsilon)} \times
 \left[\left(\frac12-\frac\nu4\right) \rho u + \frac\nu8 \rho^2 u^2 + \frac\nu{48} \rho^3 u^3 \right],
\nonumber \\
e(u) & =  - \ \frac{(1-\nu)^2}{1-\epsilon} \times 
\left( m_2 - \frac{m_1(2\tilde{\alpha}\rho-\alpha_2 m_1)}{\rho^2} \right)
 \times \left[ \frac {\mathcal{I}_0(u)}2 \ - \ \frac{\exp(2 \iota_\textup{c} u)}\rho \ L(\rho,-2\iota_\textup{c},u)  \right]
\label{val_e_crit} \\
&  \qquad + \ \frac{\nu (1-\nu) m_1}{2\rho^2(1-\epsilon)} \times \left( \tilde{\alpha}- \frac{\alpha_2 m_1} \rho \right)
\times \rho^2 \mathcal{I}_1(u)
\ - \ \frac{\alpha_2 \nu^2 m_1^2}{4\rho^3(1-\epsilon)} \times
 \left[\rho^2 \mathcal{I}_1(u) + \frac12 \rho^3 \mathcal{I}_2(u) + \frac1{12} \rho^4 \mathcal{I}_3(u) \right], 
\nonumber \\
g(u)
& =  - 2 \beta \kappa_{\infty} \times \frac{(1-\nu)^2}{1-\epsilon} \times
 \left(m_2 - \frac{m_1(2\tilde{\alpha}\rho-\alpha_2 m_1)}{\rho^2} \right) 
\left\{ \frac{\mathcal{I}_1(u)}2  -  
\frac1{2 \iota_\textup{c} \rho} \times
\left[\exp(2 \iota_\textup{c} u) L(\rho,-2\iota_\textup{c},u)-\ln \left( 1+\frac{\rho u}2 \right)\right] \right\}
\nonumber \\
&  \qquad +  \frac{\beta \kappa_\infty \nu (1-\nu) m_1}{2\rho^3(1-\epsilon)} \times \left( \tilde{\alpha}- \frac{\alpha_2 m_1} \rho \right)
\times \rho^3 \mathcal{I}_2(u)
 -  \frac{\beta \kappa_\infty \alpha_2 \nu^2 m_1^2}{4\rho^4(1-\epsilon)} \times
 \left[\rho^3 \mathcal{I}_2(u) + \frac13 \rho^4 \mathcal{I}_3(u) + \frac1{24} \rho^5 \mathcal{I}_4(u) \right],
\label{val_g_crit}
\end{align}
where, for $p \in \mathbb{N}$ and $u \geq 0$, $\mathcal{I}_p(u) := \exp(2 \iota_\textup{c} u) \int_0^u s^p \exp(-2 \iota_\textup{c} s) \text{d}s$,
and $\tilde{\alpha}, \alpha_2$ are defined in~\eqref{def_alpha_t_2}. 


\subsection{Determination of the optimal strategy}

The final step of the proof is to determine the strategy $X^*$ such that $\Pi(X^*)$ is a martingale, or equivalently such that $A^{X^*}$ is constant. Equations (\ref{eqn:dA2_Hawkes}) and (\ref{eqn:j_general}) yield
\begin{eqnarray}
\textup{d}A^X_t & = & \frac\rho{1-\epsilon} \textup{d}t \times
 \left[
\frac{D_t-(1-\epsilon)X_t}{2+\rho(T-t)} \ - \ D_t \ + \ k(T-t) \ \delta_t
\right]^2
\nonumber \\
& = & \frac{\rho/(1-\epsilon)}{[2+\rho(T-t)]^2} \ \textup{d}t \times
 \bigg[
(1-\epsilon) X_t \ + \ [1+\rho(T-t)] \ D_t
\ - \ [2+\rho(T-t)] \  k(T-t) \ \delta_t
\bigg]^2.
\nonumber
\end{eqnarray}
Thus, $A^{X^*}$ is constant on $(0,T)$ if, and only if
\begin{equation}
\text{a.s.} \ , \ \textup{d}t \ \text{-a.e. on} \ (0,T) \ , \quad 
(1-\epsilon) X^*_t = - \ [1+\rho(T-t)] \ D^*_t
\ + \ [2+\rho(T-t)] \  k(T-t) \ \delta_t,
\label{eqn:martingale_cond_Hawkes}
\end{equation}
where $D = D^*$ when the strategy $X^*$ is used by the strategic trader. Then, we characterize the strategy $X^*$ on $[0,T]$ with the three following steps:
\begin{itemize}
\item The initial jump $\Delta X^*_0$ of the strategy is such that $(X^*,D^*)$ satisfies equation (\ref{eqn:martingale_cond_Hawkes}) at time $t = 0^+$.
\item The strategy $X^*$ on $(0,T)$  is obtained by differentiating equation (\ref{eqn:martingale_cond_Hawkes}).
\item The final jump $\Delta X^*_T = -X^*_T$ closes the position of the strategic trader at time $T$.
\end{itemize}
%

%

We need the following lemma in the sequel.

\begin{lemma}\label{lemma:int_phi_delta}
Let $\phi : [0,T] \rightarrow \mathbb{R}$ be a measurable function, and for $0 \leq s \leq t \leq T$,
$ \Phi(s,t) := \int_s^t \phi(u) \exp(-\beta u) \ \textup{d}u$. We then have for all $t \in [0,T]$
%
%
\begin{eqnarray}
\int_0^t \phi(u) \ \delta_u \ \textup{d} u  
& = & \delta_0 \ \Phi(0,t) 
\ + \ \Theta_{\chi_t} \ \Phi\left({\tau_{\chi_t}},t\right)
\ + \  \overset{\chi_t-1}{\underset{i=1}{\sum}}
\Theta_i \ \Phi(\tau_i,\tau_{i+1})
\nonumber
\end{eqnarray}
\end{lemma}

\begin{proof}
The proof is straightforward since for 
$u \in [\chi_t,t], \quad \delta_u = \delta_0 \ \exp(-\beta u) + \exp(-\beta u) \ \Theta_{\chi_t}$
 and for $i \in \{0,\cdots,\chi_t-1\}$ and $u \in [\tau_i,\tau_{i+1})$, 
$\delta_u = \delta_0 \ \exp(-\beta u) + \exp(-\beta u) \ \Theta_i$.
%
%
\end{proof}

To determine the optimal strategy, only the function $k$ given by~\eqref{eqn:k_aneqb} comes into play, thus the cases $\eta = 0$ and $\eta \neq 0$ can be treated simultaneously. We also note that 
\begin{equation}
\frac{\dd}{\dd u}[u^2 \ \omega(\eta u)] = u \zeta(\eta u)
\quad \text{ and } \quad
 \frac{\dd}{\dd u}[u \ \zeta(\eta u)] = \exp(-\eta u) 
\nonumber
\end{equation}
hold for for all $u \geq 0$ and $\eta \in \mathbb{R}$.
We use Equations~\eqref{eqn:k_aneqb} and~\eqref{eqn:martingale_cond_Hawkes} to obtain the following characterization of the strategy $X^*$: a.s., $\textup{d}t$-a.e. on $(0,T)$,
\begin{equation}
(1-\epsilon) X^*_t = - \ [1+\rho(T-t)] \ D^*_t
\ + \ \frac{m_1}{2\rho} \times \left[2+\rho (T-t) \times \left\{ 1+ \zeta(\eta (T-t))+ \nu \rho (T-t) \ \omega(\eta(T-t)) \right\}\right] \ \delta_t.
\label{eqn:martingale_cond_Hawkes_noncrit}
\end{equation}
The initial jump of $X^*$ at $t=0$ is such that (\ref{eqn:martingale_cond_Hawkes_noncrit}) is verified for $t=0^+$:
\begin{equation}
(1-\epsilon) (x_0 + \Delta X^*_0)  =  -  [1 + \rho T ] \left( D_0 + (1-\epsilon) \Delta X^*_0 \right)
 +  \frac{m_1}{2\rho} \times \left[2+\rho T \times \left\{ 1+ \zeta(\eta T)+ \nu \rho T \ \omega(\eta T) \right\}\right] \ \delta_0,
\end{equation}
which gives the initial trade at time $0$ as given in Appendix~\ref{appendix:formulas_opt_strat}.

We differentiate Equation (\ref{eqn:martingale_cond_Hawkes_noncrit}) to get
\begin{eqnarray}
(1-\epsilon) \textup{d}X^*_t & = & \rho D^*_t \textup{d}t 
\ - \ [1 + \rho (T-t) ] \ \textup{d} D^*_t
\ - \ \frac{m_1}2 \times \left[ 1 +\exp(-\eta(T-t))+\nu \rho (T-t) \zeta(\eta(T-t)) \right] \ \delta_t \ \textup{d}t
\nonumber \\ 
& & \qquad + \ \frac{m_1}{2\rho} \times \left[2+\rho (T-t) \times \left\{ 1+ \zeta(\eta (T-t))+ \nu \rho (T-t) \ \omega(\eta(T-t)) \right\}\right] \ \textup{d}\delta_t.
\nonumber
\end{eqnarray}
This yields, using $\textup{d}\delta_t = -\beta \ \delta_t \ \textup{d}t + \textup{d}I_t$,
%
%
\begin{eqnarray}
(1-\epsilon) \textup{d} X^*_t
& = & \rho D^*_t \ \textup{d}t 
\ - \ m_1 \ \phi_\eta(t) \ \delta_t \ \textup{d}t
\ + \  \frac{1 + \rho (T-t)}{2 + \rho (T-t)}
 \left\{\frac {m_1} \rho \ \textup{d} I_t - (1-\nu) \ \textup{d} N_t \right\}
\label{intermed_Xstar_cont} \\
& & + \ \frac{m_1}{2\rho} \times \frac{\rho (T-t) \times \left\{\zeta(\eta (T-t)) - 1 + \nu \rho (T-t) \ \omega(\eta(T-t)) \right\}}{2+\rho(T-t)} \ \textup{d}I_t,
\nonumber
\end{eqnarray}
where for $t \in [0,T]$
\small
\begin{equation}
\phi_\eta(t) \ := \ \frac12 \times 
\frac {1 +\exp(-\eta(T-t))+\nu \rho (T-t) \zeta(\eta(T-t)) + \frac\beta\rho
 \left[2+\rho (T-t) \times \left\{ 1+ \zeta(\eta (T-t))+ \nu \rho (T-t) \ \omega(\eta(T-t)) \right\}\right]}
{2 + \rho (T-t)}
\nonumber
\end{equation}
\normalsize
and
$\delta_t = \delta_0 \ \exp(-\beta t) + \underset{0<\tau\leq t}{\sum} \exp(-\beta (t-\tau)) \ \Delta I_\tau$.
For $t \in (0,T)$,
\begin{eqnarray}
\textup{d} D^*_t  & = &  - \rho D^*_t \textup{d} t \ + \ (1-\epsilon) \textup{d} X^*_t \ + \ (1-\nu) \textup{d} N_t \nonumber \\
& = & \ - \ m_1 \ \phi_\eta(t) \ \delta_t \ \textup{d}t
\nonumber \\ 
& & \qquad + \  \frac{(1-\nu) \ \textup{d} N_t}{2 + \rho (T-t)}
\ + \ \frac{m_1}{2\rho} \times \frac{2+\rho (T-t) \times \left\{ 1+ \zeta(\eta (T-t))+ \nu \rho (T-t) \ \omega(\eta(T-t)) \right\}}{2+\rho(T-t)} \ \textup{d}I_t,
\nonumber
\end{eqnarray}
and we have
\begin{eqnarray}
D^*_{0^+} & = & D_0 + (1-\epsilon) \Delta X^*_0 \ = \ \frac {D_0 - (1-\epsilon) x_0} {2 + \rho T}
\ + \ \frac{m_1}{2\rho} \times \frac {2+\rho T \times \left\{ 1+ \zeta(\eta T)+ \nu \rho T \ \omega(\eta T) \right\}} {2 + \rho T} \ \delta_0
\nonumber \\
& & \nonumber \\
\int_{(0,t]} \textup{d} D^*_u & = & 
- \ m_1 \int_{(0,t]} \phi_\eta(u) \ \delta_u \ \textup{d} u
\ + \ \underset{0<\tau\leq t}{\sum} \frac{(1-\nu) \ \Delta N_\tau}{2 + \rho (T-\tau)}
\nonumber \\ 
& & \qquad + \ \frac{m_1}{2\rho} \times \underset{0<\tau\leq t}{\sum} 
\frac{2+\rho (T-\tau) \times \left\{ 1+ \zeta(\eta (T-\tau))+ \nu \rho (T-\tau) \ \omega(\eta(T-\tau)) \right\}}{2+\rho(T-\tau)}\ \Delta I_\tau.
\nonumber
\end{eqnarray}
We define
$ \ \Phi_\eta(s,t) := \int_s^t \phi_\eta(u) \exp(-\beta u) \ \textup{d}u \ $ for $0 \leq s \leq t \leq T$.
Lemma \ref{lemma:int_phi_delta} yields for $t \in [0,T]$
%
%
\begin{equation}
\int_0^t \phi_\eta(u) \ \delta_u \ \textup{d} u  
\ = \  \delta_0 \ \Phi_\eta(0,t) 
\ + \ \Theta_{\chi_t} \ \Phi_\eta\left({\tau_{\chi_t}},t\right)
\ + \ \overset{\chi_t-1}{\underset{i=1}{\sum}}
\Theta_i \ \Phi_\eta(\tau_i,\tau_{i+1}).
\nonumber
\end{equation}
We obtain the expression of $D^*_t$ for $t \in (0,T)$
\begin{eqnarray}
D^*_t & = & \frac {D_0 - (1-\epsilon)x_0} {2+\rho T} \ + \
\frac {\delta_0 m_1}{2\rho} \times \left[ \frac {2+\rho T \times \left\{ 1+ \zeta(\eta T)+ \nu \rho T \ \omega(\eta T) \right\}} {2 + \rho T}
\ - \ 2 \rho \ \Phi_\eta(0,t) \right] \nonumber \\
& & \quad - \ m_1 \left[ \Theta_{\chi_t} \ \Phi_\eta\left({\tau_{\chi_t}},t\right)
\ + \ \overset{\chi_t-1}{\underset{i=1}{\sum}}
\Theta_i \ \Phi_\eta(\tau_i,\tau_{i+1})  \right]
\ + \ \underset{0<\tau\leq t}{\sum} \frac{(1-\nu) \ \Delta N_\tau}{2 + \rho (T-\tau)} 
\nonumber \\
& & \quad 
+ \ \frac{m_1}{2\rho} \times \underset{0<\tau\leq t}{\sum} 
\frac{2+\rho (T-\tau) \times \left\{ 1+ \zeta(\eta (T-\tau))+ \nu \rho (T-\tau) \ \omega(\eta(T-\tau)) \right\}}{2+\rho(T-\tau)}\ \Delta I_\tau.
\nonumber
\end{eqnarray}
From~\eqref{intermed_Xstar_cont}, the strategy $X^*$ on $(0,T)$ is as given in Appendix~\ref{appendix:formulas_opt_strat}. By using again~\eqref{eqn:martingale_cond_Hawkes_noncrit}, we also get the final trade at time $T$.

We determine the function $\Phi_\eta$ in the case $\eta \neq 0$ (similar and simpler calculations yield the result for $\eta = 0$). We write
\begin{eqnarray}
\exp(-\eta(T-t)) \times \exp(- \beta t) &= & \exp(-\beta T) \times \exp(\alpha(T-t)),
\nonumber \\
(T-t) \zeta(\eta(T-t)) \times \exp(- \beta t) & = & \frac{\exp(-\beta T)}\eta \times [\exp(\beta (T-t)) - \exp(\alpha(T-t))].
\nonumber
\end{eqnarray}
Thus,
$
\phi_\eta(t) \times \exp(\beta (T-t))
$
is equal to
\small
\begin{equation}
\frac\beta2 \left( \frac1\rho+\frac\nu\eta\right) \times \exp(\beta(T-t))
\ + \ \left[ \frac12+\frac{\nu(\rho-2\beta)}{2\eta}+\frac\beta{2\eta}\left(1-\frac{\nu\rho}\eta\right) \right] \frac{\exp(\beta(T-t))}{2+\rho(T-t)}
\ + \ \left[ \frac12-\frac{\nu\rho}{2\eta}-\frac\beta{2\eta}\left(1-\frac{\nu\rho}\eta\right) \right] \frac{\exp(\alpha(T-t))}{2+\rho(T-t)},
\nonumber
\end{equation}
\normalsize
which yields for $0 \leq s \leq t \leq T$,
\begin{eqnarray}
\Phi_\eta(s,t) & = &
\frac12 \left( \frac1\rho+\frac\nu\eta\right) \times [\exp(-\beta s)-\exp(-\beta t)]
\nonumber \\
& & \quad + \ \frac{\exp(-\beta T)}{2\rho} \times
\left[ 1+\frac{\nu(\rho-2\beta)}\eta+\frac\beta\eta\left(1-\frac{\nu\rho}\eta\right) \right] \times
[L(\rho,\beta,T-s)-L(\rho,\beta,T-t) ]
\nonumber \\
& & \quad + \ \frac{\exp(-\beta T)}{2\rho} \times
\left[ 1-\frac{\nu\rho}\eta-\frac\beta\eta\left(1-\frac{\nu\rho}\eta\right) \right] \times
[L(\rho,\alpha,T-s)-L(\rho,\alpha,T-t) ].
\nonumber
\end{eqnarray}
with $\eta = \beta - \alpha \neq 0$.

\section{Proof of Theorem~\ref{Thm_Pmg}}\label{appendix:proof_Thm_Pmg} Let $X$ be an admissible strategy. We introduce the following processes: $S^N_t=S_0+\frac \nu q (N_t-N_0)$, $S^X_t=\frac \epsilon q (X_t-X_0)$,
$$ \textup{d}D^N_t  = - \rho  D^N_t  \textup{d}t
 +  \frac {1-\nu}{q}  \textup{d}N_t \text{ and }    \textup{d}D^X_t  = - \rho  D^X_t  \textup{d}t
 +  \frac {1-\epsilon}{q}  \textup{d}X_t,$$
with $D^N_0=D_0$ and $D^X_0=0$. Thus, we have $S=S^N+S^X$, $D=D^N+D^X$ and thus $P=P^N+P^X$, where $P^N=S^N+D^N$ and $P^X=S^X+D^X$. From~\eqref{CostX}, we have 
$$C(X)=\int_{[0,T)} P^N_u \ \textup{d}X_u - P_T^N X_T + C^{\text{OW}}(X),$$
where 
$$C^{\text{OW}}(X)= \int_{[0,T)} P^X_u \ \textup{d}X_u 
\ + \ \frac1{2q} \underset{0\leq\tau<T}{\sum} (\Delta X_\tau)^2
 \ - \ P_T^X X_T \ + \ \frac1{2q} \ X_T^2$$
is a deterministic function of~$X$ that corresponds to the cost when $N \equiv 0$, which is the Obizhaeva and Wang model.
We now make an integration by parts as in Remark~\ref{model_plus_mg} and get that
$$\int_{[0,T)} P^N_u \ \textup{d}X_u - P_T^N X_T =-\int_{[0,T)} X_u \ \textup{d}P^N_u.$$
When  $P^N$ is a martingale, this term has a null expectation.
Therefore, the optimal execution strategy is the same as in the Obizhaeva and Wang model, see Gatheral, Schied and Slynko~\cite{GSS}, Example 2.12, and there is no PMS.
Otherwise, we can find $0 \le s <t\le T$ such that $\E[P^N_t|\mathcal{F}_s]$ and $P^N_s$ are not almost surely equal. In this case, we consider the strategy $X_u= \E[P^N_t-P^N_s|\mathcal{F}_s] \mathbf{1}_{u\in (s,t]}$ that is a round-trip, i.e. $X_0=X_{T+}=0$. We then get
$$\E \left[-\int_{[0,T)} X_u \ \textup{d}P^N_u \right]=-\E[ (P^N_t-P^N_s) \E[P^N_t-P^N_s|\mathcal{F}_s]]=-\E[ \E[P^N_t-P^N_s|\mathcal{F}_s]^2]<0.$$
Since $C^{\text{OW}}(cX)=c^2C^{\text{OW}}(X)$, we can find $c$ small enough such that $E[C(cX)]=-c\E[ \E[P^N_t-P^N_s|\mathcal{F}_s]^2]+c^2C^{\text{OW}}(X)<0$, and therefore $cX$ is a PMS.

\bibliography{ref_Hawkes}

\begin{thebibliography}{10}

\bibitem{AbergelJedidi}
Fr\'ed\'eric Abergel and Aymen Jedidi.
\newblock A mathematical approach to order book modeling.
\newblock {\em International Journal of Theoretical and Applied Finance
  (IJTAF)}, 16(05), 2013.

\bibitem{AFS}
Aur{\'e}lien Alfonsi, Antje Fruth, and Alexander Schied.
\newblock Optimal execution strategies in limit order books with general shape
  functions.
\newblock {\em Quant. Finance}, 10(2):143--157, 2010.

\bibitem{AS_SICON}
Aur{\'e}lien Alfonsi and Alexander Schied.
\newblock Capacitary measures for completely monotone kernels via singular
  control.
\newblock {\em SIAM J. Control Optim.}, 51(2):1758--1780, 2013.

\bibitem{ASS}
Aur{\'e}lien Alfonsi, Alexander Schied, and Alla Slynko.
\newblock {Order Book Resilience, Price Manipulation, and the Positive
  Portfolio Problem}.
\newblock {\em SSRN eLibrary}, 2011.

\bibitem{AC}
Robert Almgren and Neil Chriss.
\newblock Optimal execution of portfolio transactions.
\newblock {\em Journal of Risk}, 3:5--39, 2000.

\bibitem{BDHM}
E.~Bacry, S.~Delattre, M.~Hoffmann, and J.~F. Muzy.
\newblock Modelling microstructure noise with mutually exciting point
  processes.
\newblock {\em Quant. Finance}, 13(1):65--77, 2013.

\bibitem{BDHM2}
E.~Bacry, S.~Delattre, M.~Hoffmann, and J.~F. Muzy.
\newblock Some limit theorems for {H}awkes processes and application to
  financial statistics.
\newblock {\em Stochastic Process. Appl.}, 123(7):2475--2499, 2013.

\bibitem{BacryMuzy}
E.~{Bacry} and J.~F {Muzy}.
\newblock {Hawkes model for price and trades high-frequency dynamics}.
\newblock {\em ArXiv e-prints}, January 2013.

\bibitem{BayraktarLudkovski}
Erhan Bayraktar and Michael Ludkovski.
\newblock Optimal trade execution in illiquid markets.
\newblock {\em Math. Finance}, 21(4):681--701, 2011.

\bibitem{BL}
Dimitris Bertsimas and Andrew Lo.
\newblock Optimal control of execution costs.
\newblock {\em Journal of Financial Markets}, 1:1--50, 1998.

\bibitem{BGPW}
Jean-Philippe Bouchaud, Yuval Gefen, Marc Potters, and Matthieu Wyart.
\newblock Fluctuations and response in financial markets: the subtle nature of
  "random" price changes.
\newblock {\em Quantitative Finance}, 4(2):176--190, 2004.

\bibitem{BremaudMassoulie}
Pierre Br{\'e}maud and Laurent Massouli{\'e}.
\newblock Stability of nonlinear {H}awkes processes.
\newblock {\em Ann. Probab.}, 24(3):1563--1588, 1996.

\bibitem{ContLarrard}
R.~Cont and A.~de~Larrard.
\newblock Price dynamics in a markovian limit order market.
\newblock {\em SIAM Journal on Financial Mathematics}, 4(1):1--25, 2013.

\bibitem{FonsecaZaatour}
Jos\'e Da~Fonseca and Riadh Zaatour.
\newblock Hawkes process: Fast calibration, application to trade clustering,
  and diffusive limit.
\newblock {\em Journal of Futures Markets}, pages n/a--n/a, 2013.

\bibitem{DVJ}
D.~J. Daley and D.~Vere-Jones.
\newblock {\em An introduction to the theory of point processes. {V}ol. {I}}.
\newblock Probability and its Applications (New York). Springer-Verlag, New
  York, second edition, 2003.
\newblock Elementary theory and methods.

\bibitem{Donier}
Jonathan Donier.
\newblock Market impact with autocorrelated order flow under perfect
  competition.
\newblock Papers, arXiv.org, 2012.

\bibitem{EBK}
Zolt{\'a}n Eisler, Jean-Philippe Bouchaud, and Julien Kockelkoren.
\newblock The price impact of order book events: market orders, limit orders
  and cancellations.
\newblock {\em Quant. Finance}, 12(9):1395--1419, 2012.

\bibitem{ELL}
Paul Embrechts, Thomas Liniger, and Lu~Lin.
\newblock Multivariate {H}awkes processes: an application to financial data.
\newblock {\em J. Appl. Probab.}, 48A(New frontiers in applied probability: a
  Festschrift for Soren Asmussen):367--378, 2011.

\bibitem{FGLW}
J~Doyne Farmer, Austin Gerig, Fabrizio Lillo, and Henri Waelbroeck.
\newblock How efficiency shapes market impact.
\newblock {\em Quantitative Finance}, 13(11):1743--1758, 2013.

\bibitem{FilimonovSornette}
Vladimir Filimonov and Didier Sornette.
\newblock Quantifying reflexivity in financial markets: Toward a prediction of
  flash crashes.
\newblock {\em Phys. Rev. E}, 85:056108, May 2012.

\bibitem{GDKB}
A.~Gareche, G.~Disdier, J.~Kockelkoren, and Jean-Philippe Bouchaud.
\newblock {A Fokker-Planck description for the queue dynamics of large tick
  stocks}, April 2013.

\bibitem{Gatheral}
Jim Gatheral.
\newblock No-dynamic-arbitrage and market impact.
\newblock {\em Quant. Finance}, 10(7):749--759, 2010.

\bibitem{GSS}
Jim Gatheral, Alexander Schied, and Alla Slynko.
\newblock Transient linear price impact and {F}redholm integral equations.
\newblock {\em Math. Finance}, 22(3):445--474, 2012.

\bibitem{Gueant}
Olivier Gu\'eant.
\newblock {Optimal execution and block trade pricing: a general framework}.
\newblock Papers 1210.6372, arXiv.org, October 2012.

\bibitem{HBB}
Stephen Hardiman, Nicolas Bercot, and Jean-Philippe Bouchaud.
\newblock Critical reflexivity in financial markets: a hawkes process analysis.
\newblock {\em The European Physical Journal B - Condensed Matter and Complex
  Systems}, 86(10):1--9, 2013.

\bibitem{HawkesMarks}
Alan~G. Hawkes.
\newblock Spectra of some self-exciting and mutually exciting point processes.
\newblock {\em Biometrika}, 58(1):pp. 83--90, 1971.

\bibitem{HawkesOakes}
Alan~G. Hawkes and David Oakes.
\newblock A cluster process representation of a self-exciting process.
\newblock {\em J. Appl. Probability}, 11:493--503, 1974.

\bibitem{HLR}
Weibing Huang, Charles-Albert Lehalle, and Mathieu Rosenbaum.
\newblock Simulating and analyzing order book data: The queue-reactive model.
\newblock Papers, arXiv.org, 2013.

\bibitem{HS}
Gur Huberman and Werner Stanzl.
\newblock Price manipulation and quasi-arbitrage.
\newblock {\em Econometrica}, 72(4):1247--1275, 2004.

\bibitem{JaissonRosenbaum}
Thibault Jaisson and Mathieu Rosenbaum.
\newblock {Limit theorems for nearly unstable Hawkes processes}, October 2013.

\bibitem{MTB}
Iacopo Mastromatteo, Bence Toth, and Jean-Philippe Bouchaud.
\newblock Agent-based models for latent liquidity and concave price impact.
\newblock Papers, arXiv.org, 2013.

\bibitem{OW}
Anna Obizhaeva and Jiang Wang.
\newblock Optimal trading strategy and supply/demand dynamics.
\newblock {\em Journal of Financial Markets}, 16:1--32, 2013.

\bibitem{PB}
Marc Potters and Jean-Philippe Bouchaud.
\newblock More statistical properties of order books and price impact.
\newblock {\em Physica A: Statistical Mechanics and its Applications},
  324(1-2):133--140, 2003.
\newblock Proceedings of the International Econophysics Conference.

\bibitem{PSS}
Silviu Predoiu, Gennady Shaikhet, and Steven Shreve.
\newblock Optimal execution in a general one-sided limit-order book.
\newblock {\em SIAM J. Financial Math.}, 2:183--212, 2011.

\bibitem{RobertRosenbaum}
Christian~Y. Robert and Mathieu Rosenbaum.
\newblock A new approach for the dynamics of ultra-high-frequency data: The
  model with uncertainty zones.
\newblock {\em Journal of Financial Econometrics}, 9(2):344--366, 2011.

\bibitem{StoikovWaeber}
Sasha Stoikov and Rolf Waeber.
\newblock {Optimal Asset Liquidation Using Limit Order Book Information}.
\newblock {\em SSRN eLibrary}, 2012.

\bibitem{TPLF}
Bence Toth, Imon Palit, Fabrizio Lillo, and J.~Doyne Farmer.
\newblock Why is order flow so persistent?
\newblock Papers, arXiv.org, 2011.

\bibitem{ZRA}
B.~Zheng, F.~Roueff, and F.~Abergel.
\newblock Modelling bid and ask prices using constrained hawkes processes:
  Ergodicity and scaling limit.
\newblock {\em SIAM Journal on Financial Mathematics}, 5(1):99--136, 2014.

\end{thebibliography}
\bibliographystyle{plain}

\end{document}